\documentclass[11pt]{article}
\usepackage[letterpaper,left=1in,top=1in,right=1in,bottom=1in]{geometry}

\usepackage[utf8]{inputenc}
\usepackage{ amssymb }
\usepackage{amsmath}
\usepackage{amsthm}
\usepackage{enumitem}
\usepackage{algorithm}
\usepackage[noend]{algpseudocode}
\usepackage{mathtools}
\usepackage{url}

\usepackage{authblk}
\usepackage[numbers, sort&compress]{natbib}
\usepackage[colorlinks]{hyperref}
\hypersetup{
	citecolor=blue,
   linkcolor=red,
}
\usepackage[font=scriptsize,width=.85\textwidth,labelfont=bf]{caption}
\usepackage{mathptmx}

\newtheorem{theorem}{Theorem}[section]
\newtheorem{proposition}[theorem]{Proposition}
\newtheorem{lemma}[theorem]{Lemma}
\newtheorem{definition}{Definition}[section]

\newtheorem{remark}{Remark}

\newenvironment{owndesc}%
    {\begin{description}[leftmargin = 0.2cm, labelsep = 0.2cm]}
    {\end{description}}

\newcommand{\hmu}{\hat{\mu}}
\newcommand{\leaf}{\odot}
\newcommand{\spec}{\mathfrak{s}}
\newcommand{\dup}{\mathfrak{d}}
\newcommand{\transfer}{\mathfrak{t}}
\newcommand{\child}{\mathrm{ch}}
\newcommand{\tredge}{\mathcal{E}}
\renewcommand{\hat}{\widehat}
\renewcommand{\tilde}{\widetilde}

\renewcommand{\S}{\RT }
\renewcommand{\L}{L}
\newcommand{\ats}[2]{A({#1},{#2})}
\newcommand{\hs}{\hat{S}}

\newcommand{\M}{\mathcal{M}}
\newcommand{\sT}{\sigma_{\Th}}
\newcommand{\Th}{\ensuremath{T_{\mathcal{\overline{E}}}}}
\newcommand{\RT}{\ensuremath{\mathcal{R}}}

\newcommand{\lca}{\ensuremath{\operatorname{lca}}}

\usepackage{color}

\providecommand{\keywords}[1]{\textbf{\textit{Keywords: }} #1}

\makeindex             

\title{Reconstruction of time-consistent species trees} 
\author[1]{Manuel Lafond} 
\author[2,3]{Marc Hellmuth}

\affil[1]{Department of Computer Science, Universit\'e de Sherbrooke, 2500 Boul. de l'Universit\'e, Canada \\ 	
	Email: \texttt{manuel.lafond@USherbrooke.ca}}

\affil[2]{Institute	 of Mathematics and Computer Science, University of Greifswald, Walther-
  Rathenau-Strasse 47, D-17487 Greifswald, Germany  \\ 	
	Email: \texttt{mhellmuth@mailbox.org}}

\affil[3]{
	Saarland University, Center for Bioinformatics, Building E 2.1, P.O.\ Box 151150, D-66041 Saarbr{\"u}cken, Germany
	  }

\begin{document}

\date{\ }

\maketitle

\abstract{ 
The history of gene families -- which are equivalent to event-labeled gene trees -- can to some extent be reconstructed
		from empirically estimated evolutionary event-relations containing pairs of orthologous, paralogous or xenologous
    genes. The question then arises as whether inferred event-labeled gene trees are ``biologically feasible'' which is the case
	 	if one can find a species tree with which the gene tree can be reconciled in a time-consistent way. 

		In this contribution, we consider event-labeled gene trees that contain speciation, duplication as well as horizontal gene transfer
		and we assume that the species tree is unknown. We provide a cubic-time algorithm to decide whether a ``time-consistent'' binary species 
		for a given event-labeled gene tree exists and, in the affirmative case, to construct the species tree within the same time-complexity. 		
}
\smallskip

\noindent

\sloppy

\section{Introduction}

Genes collectively  form  the  organism's  genomes and
can be viewed as ``atomic'' units whose evolutionary history forms
a tree. The  history of species, which is also a tree, and the history
of their genes  is intimately linked, since 
the gene trees evolve along the species tree.      
A detailed evolutionary scenario, therefore, consists of a gene tree,  a species tree and
a reconciliation map $\mu$ that describes how the gene tree is embedded into the species tree. 

A reconciliation map assigns vertices of the gene tree to the vertices
or edges in the species in such a way that (partial) ancestor relations given 
by the genes are preserved by the map $\mu$. This gives rise to 
three important events that may act on the genes through evolution:
\emph{speciation}, \emph{duplication}, and \emph{horizontal gene transfer (HGT)} \cite{Gray:83,Fitch2000}. 
Inner vertices of the species tree represent speciation events. 
Hence, vertices of the gene tree that are mapped to inner vertices in the species tree
underly a speciation event and are transmitted from the parent species into the daughter species. 
If two copies from a single ancestral gene are formed and reside in the same species, 
then a duplication event happened. Contrary, if one of the copies of a gene 
``jumps'' into a different branch of the species tree, then a HGT event happened. 
Since both, speciation and duplication events,  occur in between different speciation events,
such vertices of the gene trees are usually mapped to the edges of the species tree. 
The events speciation, duplication, and HGT classify pairs of genes as orthologs,
paralogs and xenologs, respectively \cite{Fitch2000}.
Intriguingly, these relations can
be estimated directly from sequence data using a variety of algorithmic
approaches that are based on the pairwise best match criterion \cite{GSH:19,GCG+18} and hence do
not require any \emph{a priori} knowledge of the topology of either the
gene tree or the species tree, see e.g.\
\cite{Roth:08,Altenhoff:09,Lechner:14,Altenhoff:16,Altenhoff2012, Lechner:11a,CBRC:02,RSLD15,Dessimoz2008,LH:92,tao2018novel,VV:19}.
Moreover, empirical estimated event-relations 
can then be used to infer the history of event-labeled gene trees 
\cite{lafond2015orthology,dondi2017approximating,LDEM:16,DEML:16,DONDI17,HHH+13,HSW:16,GAS+:17,GHLS:17} and, in some cases, also the species trees 
\cite{HW:16book,hellmuth2017biologically,HHH+12}.
This line of research, in particular, has been very successful for the reconstruction
of event-labeled gene trees and species trees based solely on the information of 
orthologous and paralogous gene pairs \cite{HLS+15}. 

In this paper, we assume that the gene tree $T$ and and the types of
evolutionary events on $T$ are known. 
For an event-labeled gene tree to be biologically feasible
there must be a putative ``true'' history that can explain the
inferred gene tree. However, in practice it is not possible to
observe the entire evolutionary history as e.g.\ gene losses
eradicate the entire information on parts of the history.
Therefore, the problem of determining whether an event-labeled gene tree is biologically feasible is reduced to the
problem of finding a valid reconciliation map, also known
as DTL-scenario \cite{THL:11,BAK:12,Doyon2010}. 
The aim is then to find the unknown species tree $S$ and 
reconciliation map between $T$ and $S$, if one exists. 
Not all event-labeled gene trees $T$, however,  are
biologically feasible in the sense that that there exists a species tree $S$
such that $T$ can be reconciled with $S$. In the absence of HGT, biologically
feasibility can be characterized in terms of ``informative'' triplets (rooted binary
trees on three leaves) that are displayed by the gene trees \cite{HHH+12}. In the
presence of HGT such triplets give at least necessary conditions for a gene tree
being biologically feasible \cite{hellmuth2017biologically}.

A particular difficulty that occurs in the presence of HGT is that gene trees with 
HGT must be mapped to species trees only in such a way that genes do not travel back in time.
To be more precise, the ancestor ordering of the vertices in a species tree 
give rise to a relative timing information of the species within the species trees. Within this context, 
speciation and duplication events can be considered as a vertical evolution, that is, 
the genetic material is transfered ``forward in time''. In contrast HGT, literally yield   
horizontal evolution, that is, genetic material is transferred such that a gene and its 
transferred copy coexist. 
N{\o}jgaard et al.\ \cite{nojgaard2018time} introduced  an  axiomatic  framework  
for  time-consistent reconciliation maps and characterize for given event-labeled gene trees $T$  
and a \emph{given} species  tree $S$  whether  there  exists  a  time-consistent  reconciliation  map or not. 
This characterization resulted in an $O(|V|\log|V|)$-time algorithm to construct 
a  time-consistent  reconciliation  map  if  one  exists. 

However, one of the crucial open questions that were left open within this
context is as follows: \emph{
For a given event-labeled gene whose
internal vertices are labeled by speciation, duplication and HGT,
does there exists a polynomial-time algorithm
to reconstruct the \emph{unknown} species tree together with a
time-consistent reconciliation map, if one
exists?}
In this contribution, we show that the answer to this problem is affirmative
and provide an $O(n^3)$ time algorithm with $n$ being 
the number of leaves of $T$ that allows is to  verify whether there is a time-consistent species $S$ for the event-labeled gene tree and, 
in the affirmative case, to construct $S$.

This paper is organized as follows. We provide in Section \ref{sec:prelim} all necessary definitions. 
	Moreover, we review some of the important results on gene and species tree, reconciliation maps and 
	time-consistency that we need here. In Section \ref{sec:gtc}, we formally introduce the gene tree consistency (GTC) problem, that 
	is, to find a time-consistent species for a given event-labeled gene tree. As a main result, we will
	see that it suffices to start with a  fully unresolved species tree that can then be stepwisely extended	
	to a binary species tree to obtain a solution to the GTC problem, provided a solution exists. 
	In Section \ref{sec:algoGTC}, we provide an algorithm to solve the GTC problem. 
For the design of this algorithm, we will utilize an auxiliary directed graph $\ats{T}{S}$ 
	based on a given event-labeled gene tree $T$ and a given species tree $S$.
	This type of graph was established in \cite{nojgaard2018time}. 
	N{\o}jgaard et al.\ \cite{nojgaard2018time} showed that there is time-consistent map between $T$ and $S$ if and only if $\ats{T}{S}$ 
	is acyclic. Our algorithm either reconstructs a species tree $S$ based on 
	the informative triplets that are displayed by the gene trees and that makes this graph  $\ats{T}{S}$ eventually acyclic
	or that returns that no solution exists. 
	The strategy of our algorithm is to construct $\ats{T}{S}$ starting with $S$ being a fully unresolved species tree
	and stepwisely resolve this tree in a way that it ``agrees'' with the informative triplets and 
	reduces the cycles in  $\ats{T}{S}$.

\section{Preliminaries}
\label{sec:prelim}

\subsubsection*{Trees, extensions and triplets}

Unless stated otherwise, all graphs in this work are assumed to be directed without explicit mention.  For a graph $G$, the subgraph induced by $X \subseteq V(G)$ is denoted $G[X]$.
For a subset $Q \subseteq V(G)$, we write $G - Q = G[V(G) \setminus Q]$. We will write $(a,b)$ and $ab$ for the edge that link $a,b\in V(G)$ of directed, resp., undirected graphs.

All trees in this work are rooted and edges are directed away from the root.
Given a tree $T$, a vertex $v \in V(T)$ is a \emph{leaf} if $v$ has out-degree
$0$, and an \emph{internal vertex} otherwise. We write $\L(T)$ to denote the set
of leaves of $T$.  A \emph{star tree} is a tree with only one internal vertex that is 
adjacent to the leaves.

We write $x \preceq_{T} y$ if $y$ lies on the unique path from 
the root to $x$, in which case $y$ is called a descendant of $x$ and $x$ is called an ancestor of $y$. We may also write $y \succeq_{T} x$ instead of $x \preceq_{T} y$. 
We use $x \prec_T y$  for $x \preceq_{T} y$ and $x \neq y$. In the latter case, 
$y$ is a \emph{strict ancestor} of $x$.
 If $x \preceq_{T} y$ or $y \preceq_{T} x$ the vertices 
$x$ and $y$ are \emph{comparable} and, otherwise, \emph{incomparable}. 
If $(x,y)$ is an edge in $T$, and thus, $y \prec_{T} x$, then $x$ is the \emph{parent}
of $y$ and $y$ the \emph{child} of $x$. We denote with $\child(x)$ the set of all children of $x$.

For a subset $X \subseteq V(T)$, the \emph{lowest common ancestor} $\lca_{T}(X)$ is the  unique $\preceq_T$-minimal vertex that is an 
ancestor of all vertices in $X$ in $T$. For simplicity, we often write $\lca_T(x,y)$ instead of $\lca_T(\{x,y\})$.

A vertex is \emph{binary} if it has $2$ children, and $T$ is
\emph{binary} if all its internal vertices are binary. A \emph{cherry} is an
internal vertex whose children are all leaves (note that a cherry may have more
than two children). A tree $T$ is \emph{almost binary} if its only non-binary
vertices are cherries. For $v \in V(T)$, we write $T(v)$ to denote the subtree
of $T$ rooted at $v$ (i.e. the tree induced by $v$ and its descendants).

\begin{definition}[Extension]
Let $x$ be a vertex of a tree $T$ with $\child(x) = \{x_1, \ldots, x_k\}$, $k\geq 3$ and suppose that  $X' \subset \child(x)$ is a strict subset of $\child(x)$. 

Then, the \emph{$(x, X')$ extension} modifies $T$ to the tree $T_{x,X'}$ as follows:  
If $|X'| \leq 1$, then put $T_{x,X'}=T$.
Otherwise, remove the edges $(x, x')$ for each $x' \in X'$ from $T$ 
		and add a new vertex $y$ together with the edges $(x,y)$ and $(y,x')$ for all $x'\in X'$ to obtain the tree $T_{x,X'}$.
\label{def:ext}
\end{definition}

Conversely, one can obtain $T$ from $T_{x,X'}$ by contracting
the edge $(x,y)$ with $y=\lca_T(X')$. Given two trees $T$ and $T'$, we say that $T'$ is a \emph{refinement} of $T$ if 
there exists a sequence of extensions that transforms $T$ into $T'$.

The \emph{restriction $T|_X$} of a tree $T$ to some subset $X \subseteq \L(T)$
is the the minimal subtree of $T$ that connects the leaves of $X$
from which all vertices with only one child have been suppressed, cf.\ \cite[Section 6.1]{semple2003phylogenetics}.

A \emph{rooted triplet}, or \emph{triplet} for short, is a binary tree with three leaves.  We write $ab|c$ to denote the unique triplet on leaf set $\{a,b,c\}$ in which the root is $\lca(a,c) = \lca(b,c)$.  We say that a tree $T$ \emph{displays} a triplet $ab|c$ if $a,b,c \in \L(T)$ and $\lca_T(a, b) \prec \lca_T(a,c) = \lca_T(b,c)$. 
We write $rt(T)$ to denote the set of rooted triplets that $T$ displays.
Given a set of triplets $R$, we say that $T$ \emph{displays} $R$ if $R
\subseteq rt(T)$. A set of triplets $R$ is \emph{compatible}, if there is a tree that displays $R$.
We also say that $T$ \emph{agrees} with $R$ if, for every
$ab|c \in R$, $ac|b \notin rt(T)$ and $bc|a \notin rt(T)$.
Note, the term ``agree'' is more general than the term ``display'' and ``compatible'', i.e., 
if $T$ displays $R$ (and thus, $R$ is compatible), then $T$ must agree with $R$. The converse, however, is not always true. 
To see this, consider the star tree $T$ , i.e., $rt(T) = \emptyset$, and let $R=\{ab|c,bc|a\}$. 
It is easy to verify that $R$ is incompatible since 
there cannot be any tree that displays both triplets in $R$. However, the set $R$ agrees with $T$.

We will consider rooted trees
$T=(V,E)$ from which particular edges are removed. Let $\tredge_T\subseteq E$ and
consider the forest $\Th\coloneqq (V,E\setminus \tredge_T)$. We can preserve the
order $\preceq_T$ for all vertices within one connected component of $\Th$ and
define $\preceq_{\Th}$ as follows: $x\preceq_{\Th}y$ iff $x\preceq_{T}y$ and
$x,y$ are in same connected component of $\Th$. Since each connected component
$T'$ of $\Th$ is a tree, the ordering $\preceq_{\Th}$ also implies a root
$\rho_{T'}$ for each $T'$, that is, $x\preceq_{\Th} \rho_{T'}$ for all $x\in
V(T')$. If $L(\Th)$ is the leaf set of $\Th$, we define $L_{\Th}(x) = \{y\in
L(\Th) \mid y\preceq_{\Th} x\}$ as the set of leaves in $\Th$ that are reachable
from $x$. Hence, all $y\in L_{\Th}(x)$ must be contained in the same connected
component of $\Th$. We say that the forest $\Th$ displays a triplet $r$, if $r$
is displayed by one of its connected components. Moreover, $rt(\Th)$ denotes
the set of all triplets that are displayed by the forest $\Th$.  
We simplify the notation a bit and write $\sT(u):=\sigma(L_{\Th}(u))$.

\subsubsection*{Gene and species trees}

Let $\Gamma$ and $\Sigma$ be a set of genes and a set of species, respectively.  Moreover, we assume to know the gene-species association, i.e., a surjective map  $\sigma: \Gamma \rightarrow \Sigma$.
A \emph{species tree} is a tree $S$ such that $\L(S) \subseteq \Sigma$.
A \emph{gene tree} is a  tree $T$ such that $\L(T) \subseteq \Gamma$. 
  Note that $\sigma(l)$ is defined for every leaf $l \in \L(T)$.  
  We extend $\sigma$ to interval vertices of $T$, and put 
$\sigma_T(v) = \{\sigma(l) : l \in \L(T(v))\}$ for an internal vertex $v$ of $T$.  We may drop the $T$ subscript whenever there is no risk of confusion.
We emphasize that species and gene trees need not to be binary.

Given a gene tree $T$, we assume knowledge of a labeling function $t : V(T) \cup
E(T) \rightarrow \{\leaf, \spec, \dup, \transfer\} \cup \{0, 1\}$. We require
that $t(v) \in \{\leaf, \spec, \dup, \transfer\}$ for all $v \in V(T)$ and $t(e)
\in \{0,1\}$ for all $e \in E(T)$. Each symbol represents a different vertex
type: $\leaf$ are leaves, $\spec$ are speciations, $\dup$ are duplications and
$\transfer$ indicates vertices from which a horizontal gene transfer started. Edges labeled by $1$ represent horizontal transfers
and edges labeled by $0$ represent vertical descent. 
Here, we always assume that only edges $(x,  y)$ for which $t(x) = \transfer$ might be
labeled as transfer edge; $t(x,y)=1$. We let $\tredge_T
= \{e \in E(T) \colon t(e) = 1\}$ be the set of transfer edges.
For technical reasons, we
also require that $t(u) = \leaf$ if and only if $u \in \L(T)$. 

We write $(T; t, \sigma)$ to denote a gene tree $T$ labeled by $t$ having gene-species mapping $\sigma$.

In what follows we will only consider labeled gene trees $(T;t,\sigma)$ that satisfy the following three axioms:
\begin{description}
\item[(O1)] Every internal vertex $v$ has outdegree at least $2$. 
\item[(O2)] Every transfer vertex $x$ has at least one transfer edge $e=(x,v)$ labeled $t(e)=1$, and at least one non-transfer edge $f=(x,w)$ labeled $t(e)=0$; 
\item[(O3)] \emph{\textbf{(a)}} 
If $x\in V$ is a speciation vertex with children $v_1,\dots,v_k$, $k\geq 2$,
then   $\sT(v_i) \cap \sT(v_j) =\emptyset$,  $1\leq i<j\leq k$;
	
\emph{\textbf{(b)}} 
If $(x,y) \in \tredge_T$, then 
	$\sT(x)\cap \sT(y) = \emptyset$.
\end{description}

These conditions are also called ``observability-axioms'' and are 
fully discussed in \cite{hellmuth2017biologically,nojgaard2018time}. 
We repeat here shortly the arguments to justify Condition (O1)-(O3). 
Usually the considered labeled gene trees are obtained from genomic sequence data. 
Condition (O1)  ensures that every inner vertex leaves a historical 
trace in the sense that there are at least two children
that have survived. If this were not the case, we would have no evidence that vertex $v$ ever exist.
Condition (O2) ensures that for an HGT event a historical trace remains of both the transferred and the non-transferred copy.
Furthermore, there is no clear evidence for a speciation vertex $v$ 
if it does not ``separate'' lineages, which is ensured by Condition (O3.a).
Finally (O3.b) is a simple consequence of the fact that if a transfer edge $(x,y)$
in the gene tree occurred, then the species $X$ and $Y$ that contain $x$ and $y$, respectively, 
cannot be ancestors of each other, as otherwise, the species $X$ and $Y$ would not coexist
(cf.\ \cite[Prop.\ 1]{nojgaard2018time}). 

We emphasize that Lemma 1 in \cite{nojgaard2018time} states that the leaf set	
		$L_1,\dots,L_k$ of the connected components $T_1,\dots,T_k$ of $\Th$
		forms a partition of $L(T)$, which directly implies that 
	$\sT(x) \neq \emptyset$ for all $x\in V(T)$.

\subsubsection*{Reconciliation maps and speciation triplets}

A \emph{reconciliation map from $(T;t,\sigma)$ to $S$} 
is a map $\mu : V(T) \rightarrow V(S)\cup E(S)$ that 
satisfies the following constraints for all $x\in V(T)$:  

\begin{description}
\item[(M1)] \emph{Leaf Constraint.}  If $x\in \Gamma$,   then $\mu(x)=\sigma(x)$. 		\vspace{0.03in}
\item[(M2)] \emph{Event Constraint.}
	\begin{itemize}
		\item[(i)]  If $t(x)=\spec$, then  
                  $\mu(x) = \lca_S(\sT(x))$.
		\item[(ii)] If $t(x) \in \{\dup, \transfer\}$, then $\mu(x)\in E(S)$. 
		\item[(iii)] If $t(x)=\transfer$ and $(x,y)\in \tredge_T$, 
						 then $\mu(x)$ and $\mu(y)$ are incomparable in $S$. 
		\item[(iv)] If $t(x)=\spec$, then  $\mu(v)$ and $\mu(u)$ are incomparable in $S$ for all distinct $u,v\in\child(x)$.	
 	\end{itemize} \vspace{0.03in}
\item[(M3)] \emph{Ancestor Constraint.}		\\	
		Let $x,y\in V$ with $x\prec_{\Th} y$. 
		Note, the latter implies that the path connecting $x$ and $y$ in $T$
		does not contain transfer edges.
		We distinguish two cases:
	\begin{itemize}
		\item[(i)] If $t(x),t(y)\in \{\dup, \transfer\}$, then $\mu(x)\preceq_S \mu(y)$, 
		\item[(ii)] otherwise, i.e., at least one of $t(x)$ and $t(y)$ is a speciation $\spec$, $\mu(x)\prec_S\mu(y)$.
	\end{itemize}
\end{description}
We call $\mu$ the \emph{reconciliation map} from $(T;t,\sigma)$ to $S$. 
The provided definition of a reconciliation map coincides with the one as given in 
in \cite{hellmuth2017biologically,nojgaard2018time,NGD+17} and is a 
natural generalization of the maps as in \cite{HHH+12,Doyon:09}
for the case, no HGT took place. In case that the event-labeling of $T$ is unknown, but a
species tree $S$ is given, the authors in  \cite{THL:11,BAK:12} gave an axiom set, 
called DTL-scenario, to reconcile $T$ with $S$. This
reconciliation is then used to infer the event-labeling $t$ of $T$.
Our axiom set for the reconciliation map is more general, nevertheless, 
equivalent to DTL-scenarios in case the considered gene trees are binary \cite{nojgaard2018time,NGD+17}.

The question arises when for a given gene tree $(T;t,\sigma)$ a species tree $S$ together with a reconciliation map $\mu$ from $(T;t,\sigma)$
to $S$ exists. An answer to this question is provided by 

\begin{definition}
Let $(T;t,\sigma)$ be a labeled gene tree 
The set $\RT (T;t,\sigma)$ is the set
of triplets $\sigma(a)\sigma(b)|\sigma(c)$ where $\sigma(a),\sigma(b),\sigma(c)$ are pairwise distinct
and either
\begin{enumerate}
	\item  $ab|c$ is a triplet displayed by $\Th$ and
					$t(\lca_{\Th} (a, b, c)) = \spec$ or 
	\item $a, b \in L(\Th (x))$ and $c \in L(\Th (y))$ for some transfer edge $(x,y)$ or $(y,x)$ in $\tredge_T$
\end{enumerate}
\label{def:informativeTriplets}
\end{definition}
 
\begin{theorem}[\cite{hellmuth2017biologically}]
 Let $(T;t,\sigma)$ be a labeled gene tree. 
Then, there is a species tree $S$ together with a reconciliation map $\mu$ from $(T;t,\sigma)$ to $S$ if and only if  $\RT (T;t,\sigma)$
is compatible. In this case, every species tree $S$ that displays
$\RT (T;t,\sigma)$ can be reconciled with  $(T;t,\sigma)$. 

Moreover, there is a polynomial-time algorithm  that returns a species tree $S$ for 
$(T;t, \sigma)$ together with a reconciliation map $\mu$ in
	polynomial time, if one exists and otherwise, returns that
	there is no species tree for $(T;t, \sigma)$.
\label{thm:SpeciesTriplets}
\end{theorem}

It has been shown in \cite{hellmuth2017biologically}, that if there is any reconciliation map 
		from $(T;t,\sigma)$	to $S$, then there is always 
 	 a reconciliation map $\mu$ that additionally satisfies for all $u\in V(T)$ with $t(u)\in \{\dup,\transfer\}$:
	 \[\mu(u) = (v,\lca_S(\sT(u))\in E(S)\] where $v$ denotes the unique
   parent of  $\lca_S(\sT(u))$ in $S$. 

\begin{definition}	
	We define a simplified map $\hmu_{T,S}\colon V(T) \to V(S)$ that associates a vertex $v \in V(T)$ to the lowest common ancestor of $\sT(v)$, i.e., 
	\[ \hmu_{T, S}(v) \coloneqq lca_{S}(\sT(v))
\] 

	We call $\hmu_{T, S}$ an \emph{LCA-map}. 
\label{def:hmu}
\end{definition}

\begin{remark}
Note that if $v$ is a leaf of $T$, we have 	$\hmu_{T, S}(v) = \sigma(v)$.
	Moreover, the LCA-map $\hmu_{T, S}$ always exists and is uniquely defined, although there might be no 
	reconciliation map from $(T;t,\sigma)$ to $S$. 
\label{rem:hmu}
\end{remark}

\begin{figure}[tbp]
	\begin{center}
		\includegraphics[width=.8\textwidth]{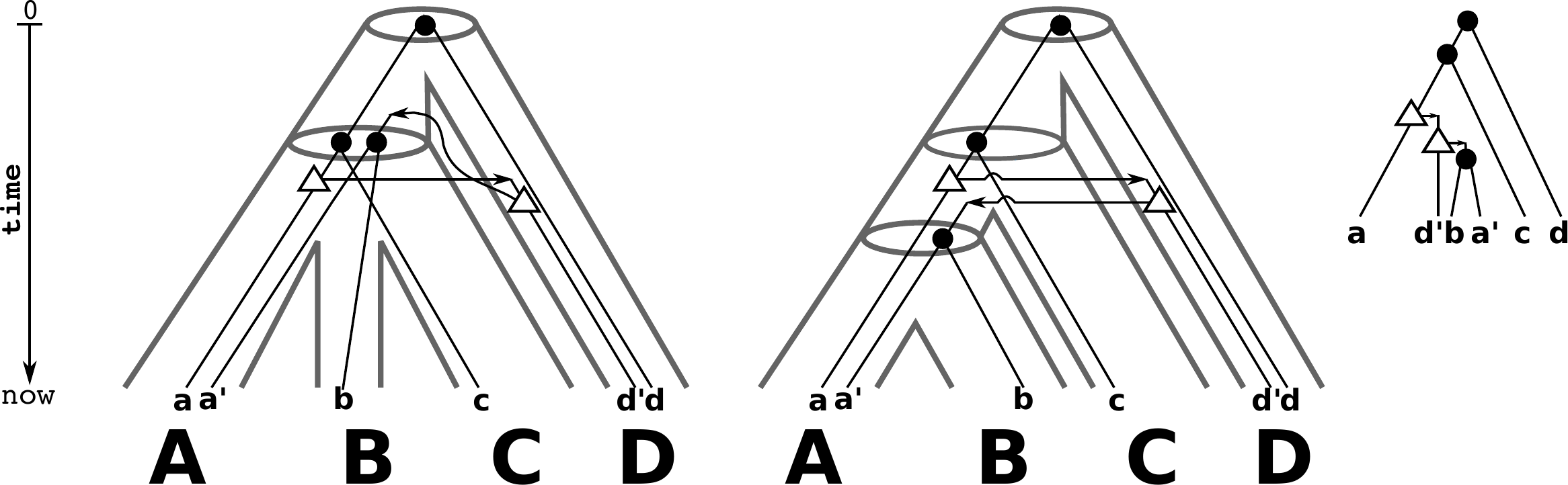}
	\end{center}
	\caption{Taken from \cite[Fig.\ 4]{hellmuth2017biologically}. 
		From the binary gene tree $(T;t,\sigma)$  (right)                                             
		we obtain the species triples $\S(T;t,\sigma) = \{AB|D,AC|D\}$.
		Note, vertices $v$ of $T$  with $t(v)=\spec$ and $t(v)=\transfer$
	  are highlighted by ``$\bullet$'' and ``$\triangle$'', respectively.
		Transfer edges are marked with an ``arrow''. 	                  
		Shown are two  (tube-like)  species trees  (left and middle) that 
		display $\S(T;t,\sigma)$. Thus, Theorem \ref{thm:SpeciesTriplets}	implies that
		for both trees a reconciliation map from $(T;t,\sigma)$ exists. 
	The respective  reconciliation maps 
		for $(T;t,\sigma)$ and the species tree are given implicitly by 
		drawing $(T;t,\sigma)$ within the species tree. 
		The left species tree $S$ is least resolved for $\S(T;t,\sigma)$. 
		The reconciliation map  from $(T;t,\sigma)$ to $S$ is unique,
		however,   not time-consistent. Thus, no time-consistent reconciliation
		between $T$ and $S$ exists at all. 
		On the other hand, for $T$ and the middle species tree (that is a refinement of $S$)
		there is a time-consistent reconciliation map. 
	}
\label{fig:least}
\end{figure}

We may write $\hmu, \hmu_{T}$ or $\hmu_{S}$ if $T$ and/or $S$ are clear from the context.

Note, however, that compatibility of  $\RT (T; t,\sigma )$ only provides a necessary condition for
a the existence of \emph{biologically feasible} reconciliation, i.e., maps that are additionally time-consistent. 
To be more precise:

\begin{definition}[Time Map]\label{def:time-map}
  The map $\tau_T\colon V(T) \to \mathbb{R}$ is a time map for the 
  rooted tree $T$ if $x\prec_T y$ implies $\tau_T(x)>\tau_T(y)$ for all 
  $x,y\in V(T)$. 
\end{definition}

\begin{definition}[Time-Consistent] \label{def:tc-mu} A reconciliation map $\mu$ from
  $(T;t,\sigma)$ to $S$ is \emph{time-consistent} if there are time maps
  $\tau_T$ for $T$ and $\tau_S$ for $S$ satisfying the
  following conditions for all $u\in V(T)$:
  \begin{description}
  \item[(B1)] If $t(u) \in \{\spec, \leaf \}$, then 
    $\tau_T(u) = \tau_S(\mu(u))$. \label{bio1}
  \item[(B2)] If $t(u)\in \{\dup,\transfer \}$ and, thus
    $\mu(u)=(x,y)\in E(S)$, \label{bio2} then
    $\tau_S(y)>\tau_T(u)>\tau_S(x)$. 
\end{description}
	If a time-consistent reconciliation map  from $(T;t,\sigma)$ to $S$ exists, 
			we also say that $S$ is a \emph{time-consistent species tree for}  $(T;t,\sigma)$.
\end{definition}

Figure \ref{fig:least} gives an example for two different species trees that 
both display $\S(T;t,\sigma)$ for which only one admits a time-consistent 
reconciliation map. Further examples can be found in 
\cite{hellmuth2017biologically,nojgaard2018time}. 
To determine whether a time-consistent map for a given gene and species tree
exists we will use an auxiliary graph as defined in \cite{nojgaard2018time}. 
We will investigate the structure of this graph in the remaining part of this section.

\subsubsection*{Auxiliary graph construction}

Let $(T;t,\sigma)$ be a labeled gene tree and $S$ be a species tree.
Let $\ats{T}{S}$ be the graph with vertex set $V(\ats{T}{S}) = V(T) \cup V(S)$, 
and edge set $E(\ats{T}{S})$ constructed from four sets as follows:
\begin{description} 
    \item[(A1)] 
    for each $(u, v) \in E(T)$, we have $(u', v') \in E(\ats{T}{S})$, where 
    
    \[
    u' = \begin{cases}
        \hmu(u) &\mbox{if $t(u) \in \{\leaf, \spec\}$} \\
        u &\mbox{otherwise}
    \end{cases}
     \quad \quad \quad
    v' = \begin{cases}
    	    \hmu(v) &\mbox{if $t(v) \in \{\leaf, \spec\}$} \\
    	    v &\mbox{otherwise}
    \end{cases}
    \]
    
    \item[(A2)]
    for each $(x, y) \in E(S)$, we have $(x, y) \in E(\ats{T}{S})$
    
    \item[(A3)]
    for each $u$ with $t(u) \in \{\dup, \transfer\}$, 
    we have $(u, \hmu(u)) \in E(\ats{T}{S})$
    
    \item[(A4)]
		  for each $(u, v) \in \tredge_{T}$, we have 
		  $(lca_S(\hmu(u), \hmu(v)), u) \in E(\ats{T}{S})$
\end{description}

We are aware of the fact that the graph $\ats{T}{S}$ heavily depends on 
the event-labeling $t$ and the species assignment $\sigma$ of the gene tree $(T;t,\sigma)$. 
However, to keep the notation simple we will write, by slight abuse of notation, $\ats{T}{S}$
instead of the more correct notation  $\ats{(T;t,\sigma)}{S}$.
The $\ats{T}{S}$ graph has four types of edges, and we shall refer to them as the A1, A2, A3-and A4-edges, respectively.  
We note for later reference that if $(x, y)$ is an A1-edge such that $x, y \in V(S)$, then we must have $y \preceq_S x$ which follows from the definition of $\hmu_{T,S}$
and the fact that $\sT(y) \subseteq \sT(x)$.

We emphasize, that the definition of  $\ats{T}{S}$ slightly differs from the 
		one provided in \cite{nojgaard2018time}. While Properties (A2), (A3) and (A4)
		are identical, (A1) was defined in terms of a reconciliation map $\mu$
		from $(T;t,\sigma)$ to $S$ in \cite{nojgaard2018time}.
		 To be more precise, in \cite{nojgaard2018time} it is stated $u' = \mu(u)$
	  and $v' = \mu(v)$ for speciation vertices or leaves $u$ and $v$ 
		instead of $u' = \hmu(u)$	  and $v' = \hmu(v)$, respectively. 
		However, Condition (M1) and (M2.i) imply that $\mu(u)=\hmu(u)$ 
		and $\mu(v)=\hmu(v)$ provided $\mu$ exists. 
		In other words, the definition of  $\ats{T}{S}$ here and in \cite{nojgaard2018time}
		are identical, in case a reconciliation map $\mu$ exists. 

		Since we do not want to restrict ourselves to the existence of a reconciliation
		map (a necessary condition is provided by Theorem \ref{thm:SpeciesTriplets})
		we generalized the definition of $\ats{T}{S}$ in	terms of  $\hmu$ instead. 
		
		For later reference, we summarize the latter observations in the following remark.

\begin{remark}
	The graph $\ats{T}{S}$ does not explicitly depend on a reconciliation map. 
	That is, even if there is no reconciliation map at all, $\ats{T}{S}$ is always
	well-defined.
	\label{rem:ats-well-defined}
\end{remark}

The next lemma deals with possible self-loops in  $\ats{T}{S}$.

\begin{lemma}
Let $(T;t,\sigma)$ be an event-labeled gene tree, $S$ be a species tree and $S^*$ be a refinement of $S$. 
Moreover, let $l$ be a leaf of $S$ (and thus, of $S^*$). 
Then $(l, l)$ is an edge of $\ats{T}{S}$ if and only if $(l, l)$ is an edge of $\ats{T}{S^*}$.

Furthermore, if there is a reconciliation map from  $(T;t,\sigma)$ to $S$
then, the graph $\ats{T}{S}$ will never contain self-loops and  
every edge $(u',v')$  in $\ats{T}{S}$ with $u',v'\in V(S)$,
is either an A1- or A2-edge and satisfies $v'\prec_S u'$.
\label{lem:self-loops}
\end{lemma}
\begin{proof}
Let $(T;t,\sigma)$ be an event-labeled gene tree, $S$ be a species tree and $S^*$ be a refinement of $S$.
Note that if $(l, l)$ is a self-loop of $\ats{T}{S}$ (respectively $\ats{T}{S^*}$), then $(l,l)$ must be an A1-edge, and so  there is $(u, v) \in E(T)$ such that $\hmu_S(u) = \hmu_S(v) = l$ (respectively $\hmu_{S^*}(u) = \hmu_{S^*}(v) = l$).  Since $l$ is a leaf, $\hmu_S(u) = \hmu_S(v) = l$ if and only if $\hmu_{S^*}(u) = \hmu_{S^*}(v) = l$, and the statement follows.

For the second statement, assume that there is a reconciliation map from  $(T;t,\sigma)$ to $S$. 
To see that $\ats{T}{S}$ does not contain self-loops, observe once again that self-loops can only be provided by A1-edges. So
assume, for contradiction, that there is an edge $(u, v) \in E(T)$ such that
$t(u),t(v)\in \{\leaf, \spec\}$ and $\hmu(u)=\hmu(v)$. Since $t(u),t(v)\in
\{\leaf, \spec\}$, Property (M1) and (M2.i) imply that $\hmu(u)=\mu(u)$ and
$\hmu(v)=\mu(v)$ for every reconciliation map $\mu$ from $(T;t,\sigma)$ to $S$.
Since $v\prec_T u$, Condition (M3.ii) implies that $\hmu(v) =
\mu(v)\prec_S\mu(u)=\hmu(u)$; a contradiction. 		

Now let  $(u',v')$ be an edge in $\ats{T}{S}$ with $u',v'\in V(S)$. 
Since all A3- or A4-edges involve vertices of $T$, we can conclude that 
 $(u',v')$ must either be an A1-edge or an A2-edge. 
Clearly, if $(u',v')$ is an A2-edge, we trivially have $v'\prec_S u'$. 
Assume that $(u',v')$ is an A1-edge. 
Hence there is an edge $(u,v)\in E(T)$ such that 
$u' = \hmu(u)$ and $v' = \hmu(v)$. This implies that
$t(u),t(v)\in \{\spec,\leaf\}$. Condition (M1) and (M2.i) imply
$\hmu(u) = \mu(u)$ and $\hmu(v) = \mu(v)$. Moreover, 
since  $(u,v)\in E(T)$, we have $t(u)=\spec$.
Now, we can apply Condition (M3.ii) to conclude that 
$v' = \hmu(v)=\mu(v) \prec_S \mu(u)  = \hmu(u) = u'$.
\end{proof}

The graph $\ats{T}{S}$ will be utilized to characterize gene-species tree pairs that admit a time-consistent reconciliation map.
For a given gene tree $(T;t,\sigma)$ and a given species tree $S$ the existence of a time-consistent reconciliation map
can easily be verified as provided by the next 

\begin{theorem}[\cite{hellmuth2017biologically,nojgaard2018time}]
Let $(T;t,\sigma)$  be a labeled gene tree and $S$ be a species tree. 
Then $T$ admits a time-consistent reconciliation map with $S$ 
if and only if $S$ displays every triplet of $\RT (T; t,\sigma )$  and $\ats{T}{S}$ is acyclic.

The time-consistent reconciliation map can then be constructed in $O(|V(T)|\log(|V(S)|))$ time. 
\label{thm:timeCons-S}
\end{theorem}

\section{Gene Tree Consistency}
\label{sec:gtc}

The main question of interest of this work is to determine whether a species tree $S$ exists at all for a labeled gene tree $T$.
Here, we solve a slightly more general problem: the one of refining a given almost binary species tree $S$ so that $T$ can be reconciled with it.

\vspace{3mm}
\noindent 
The \textsc{Gene Tree Consistency} (GTC) problem:

\noindent
\textbf{Given:} A labeled gene tree $(T;t,\sigma)$ and an almost binary species tree $S$. 	

\noindent
\textbf{Question:} Does there exist a refinement $S^*$ of $S$ that displays $\RT (T; t,\sigma)$ 
									 and such that $\ats{T}{S^*}$ is acyclic?

\vspace{3mm}

It is easy to see that the problem of determining the existence of a species
tree $S$ that displays $\RT (T; t,\sigma)$ 
and such that $\ats{T}{S}$ is acyclic is a special case of this problem. Indeed, it suffices
to provide a star tree $S$ as input to the GTC problem, since every species tree
is a refinement of $S$.

\begin{definition}
A species tree $S^*$ \emph{is a solution} to a given GTC instance $((T;t,\sigma), S)$ if $S^*$ displays $\RT (T; t,\sigma )$
and $\ats{T}{S^*}$ is acyclic.
\end{definition}

We first show that, as a particular case of the following lemma, one can restrict the search to binary species trees (even if $T$ is non-binary).

\begin{lemma}\label{lem:binarize}
Let $((T;t,\sigma), S)$ be a GTC instance and assume that a species tree $\hs$ is a
solution to this instance. Then any refinement $S^*$ of $\hs$ is also a solution
to $((T;t,\sigma), S)$.
\end{lemma}

\begin{proof}
We may assume that $\hs$ is non-binary as otherwise we are done. Let $S^*$ be
any refinement of $\hs$. First observe that we have $\RT (T;t,\sigma)\subseteq rt(\hs) \subseteq rt(S^*)$,
and thus $S^*$ displays $\RT (T;t,\sigma)$.  

It remains to show that 
$\ats{T}{S^*}$ is
acyclic. We first prove that any single $(x,X')$ extension applied to $\hs$ preserves
acyclicity. Let $S'\coloneqq \hs_{x,X'}$ be any tree obtained from $\hs$ after applying some $(x,X')$ extension. 
As specified in Definition \ref{def:ext}, if $|X'| \leq 1$, then $S'=\hs$. In this case, $\ats{T}{S'} = \ats{T}{\hs}$ is acyclic and we are done. 
Hence, suppose that $|X|>1$. Thus, a new node $y$ was created, 
added as a child of $x$ and became the new parent of $X' \subset X$. 
We claim that $\ats{T}{S'}$ is acyclic. For
the remainder, we will write $\hmu_{\hs}$ and $\hmu_{S'}$ instead of
$\hmu_{T,\hs}$ and $\hmu_{T, S'}$ since $T$ will remain fixed. We will make use
of the following properties.

\begin{enumerate} 

\item[(P1)] for every subset $Z \subseteq \L(S)$, it holds that 

	 $\lca_{\hs}(Z) = \begin{cases}
        \lca_{S'}(Z) &\mbox{if $\lca_{S'}(Z)\neq y$}        \\
        x &\mbox{otherwise}
    \end{cases}$

\item[(P2)] For every $u\in V(T)$, it holds that

	 $\hmu_{\hs}(u) = \begin{cases}
        \hmu_{S'}(u) &\mbox{if $\hmu_{S'}(u) \neq y$ \textit{(Case P2.a)}} \\
        x &\mbox{otherwise \textit{(Case P2.b)}}
    \end{cases}$

\end{enumerate}

Property (P1) follows from the fact that $\L(S'(v)) = \L(\hs(v))$ for any $v \in
V(S') \setminus \{y\}$ and $\L(S'(y)) \subset \L(\hs(x))$. Therefore if
$\lca_{S'}(Z)  = z\neq y $, then $z$ is also a common ancestor of $Z$ in
$\hs$ and there cannot be lower common ancestor below $z$. If $z=y$,
then $x$ is a common ancestor of $Z$ in $\hs$ and there cannot be a lower common
ancestor below $x$. Property (P2) is a direct consequence of (P1) and the
definition of $\hmu_{S'}$ and $\hmu_{\hs}$.

Now, suppose for contradiction that $\ats{T}{S'}$ contains a cycle $C = (w_1,
\ldots, w_k, w_1)$. 
Note that $\RT (T;t,\sigma)\subseteq rt(\hs) \subseteq rt(S')$. 
Thus,  Theorem \ref{thm:SpeciesTriplets}
implies that there is a reconciliation map from 
from $(T;t,\sigma)$ to $S'$. 
By Lemma \ref{lem:self-loops}, $\ats{T}{S'}$ does not
contain self-loops and thus $k>1$ for $C = (w_1,\ldots, w_k, w_1)$ 

Consider the sequence of vertices $\tilde{C} = (\tilde{w}_1, \ldots,
\tilde{w}_k, \tilde{w}_1)$ of vertices of $\ats{T}{\hs}$ where we take $C$, but
replace $y$ by $x$ if it occurs. That is, we define, for each $1 \leq i \leq k$:
$$
\tilde{w}_i = \begin{cases}
w_i &\mbox{ if $w_i \neq y$ }\\
x &\mbox{ if $w_i = y$ }
\end{cases}
$$

We claim that every element in $\{(\tilde{w}_1, \tilde{w}_2), \ldots, (\tilde{w}_{k-1},
\tilde{w}_k),(\tilde{w}_k, \tilde{w}_1)\} \setminus \{(x, x)\}$ is an edge in 
$\ats{T}{\hs}$ (the pair $(x, x)$ can occur in $\tilde{C}$ if $(x, y)$ is in
$C$, but we may ignore it). 
This will imply the existence of a cycle in
$\ats{T}{\hs}$, yielding a contradiction.

We show that $(\tilde{w}_1, \tilde{w}_2) \in E(\ats{T}{\hs})$, assuming that
$(\tilde{w}_1, \tilde{w}_2) \neq (x, x)$. This is sufficient to prove our claim,
since we can choose $w_1$ as any vertex of $C$ and relabel the other vertices
accordingly.

\begin{description}
\item[\textnormal{\em Case: $(w_1, w_2)$ is an A1-edge.} ] \ \\
Since $(w_1, w_2)$ is an A1-edge, it is defined by some edge $(u, v) \in E(T)$ and
must coincide with one of the edges in  $\mathcal A = \{(u, v), (u, \hmu_{S'}(v)), (\hmu_{S'}(u), v),
(\hmu_{S'}(u), \hmu_{S'}(v))\}$. 

Suppose that $w_1,w_2\neq y$. 
Then, by construction of $\tilde w_1$ and $\tilde w_2$, we have
$\tilde w_1=w_1$ and $\tilde w_2=w_2$. 
Hence, $(\tilde{w}_1, \tilde{w}_2) = (w_1, w_2)$ is an edge in  $\mathcal A$. 
By (P2), $\hmu_{\hat S}(u)=\hmu_{S'}(u)$ and $\hmu_{\hat S}(v)=\hmu_{S'}(v)$. 
Hence,  $(\tilde{w}_1, \tilde{w}_2)$ is of one of the form
$(u, v), (u, \hmu_{\hat S}(v)), (\hmu_{\hat S}(u), v), (\hmu_{\hat S}(u), \hmu_{\hat S}(v))$. 
This implies that $(\tilde{w}_1, \tilde{w}_2)$ is an A1-edge that is contained in 
$\ats{T}{\hs}$.

If $w_1 =y$, then $y\in V(S')$ implies that $y=\hmu_{S'}(u)$. 
By construction and (P2.b), $\tilde{w}_1 = x = \hmu_{\hat S}(u)$. 
This, in particular, implies that $w_2 \notin \{x, y\}$ as otherwise, $\tilde{w}_2 = x$; contradicting
$(\tilde{w}_1, \tilde{w}_2) \neq (x, x)$.
By construction of $\tilde w_2$, we have  $\tilde{w}_2=w_2$.
Thus, $(\tilde{w}_1,\tilde{w}_2)$ is either of the form 
$(\hmu_{\hat S}(u), v)$ or $(\hmu_{\hat S}(u), \hmu_{\hat S}(v))$ depending on the label $t(v)$. 
In either case, $(\tilde{w}_1,\tilde{w}_2)$ is an A1-edge that is contained in 
$\ats{T}{\hs}$. 

If $w_2 = y$ then, by analogous arguments as in the case $w_1=y$, we have  
$\tilde{w}_2 = x = \hmu_{\hat S}(v)$ and $\tilde{w}_1=w_1$. Again, 
$(\tilde{w}_1,\tilde{w}_2)$ is an A1-edge that is contained in 
$\ats{T}{\hs}$.

In summary, $(\tilde{w}_1,\tilde{w}_2)$ is an A1-edge in 
$\ats{T}{\hs}$ whenever $(w_1, w_2)$ is an A1-edge in $\ats{T}{S'}$

\item[\textnormal{\em Case: $(w_1, w_2)$ is an A3-edge.} ] \ \\
Since $(w_1, w_2)$ is an A3-edge, we have $(w_1, w_2) = (u, \hmu_{S'}(u))$. 
Since $u \in V(T)$, it holds that $w_1=u \neq y$ and thus, $\tilde{w}_1=w_1=u$.
Now we can apply similar arguments as in the first case:
either $\hmu_{S'}(u) \neq y$ and thus, $\tilde{w}_2 = w_2 = \hmu_{S'}(u) =  \hmu_{\hs}(u)$
or $\hmu_{S'}(u) =y$ and thus, $\tilde{w}_2 = x = \hmu_{\hs}(u)$.
In both cases, $(\tilde{w}_1, \tilde{w}_2) = (u,\hmu_{\hs}(u))$ which implies that 
$(\tilde{w}_1, \tilde{w}_2)$ is an A3-edge in $\ats{T}{S'}$.

\item[\textnormal{\em Case: $(w_1, w_2)$ is an A2-edge.} ] \ \\
Since $(w_1, w_2)$ is an A2-edge, we have $(w_1, w_2)\in E(S')$ and hence, 
$w_1$ is the parent of $w_2$ in $S'$. This implies that $w_2 \neq y$ as, otherwise,
$w_1=x$ and thus, $(\tilde{w}_1, \tilde{w}_2) = (x, x)$; a contradiction. 
Thus, by construction, $\tilde{w}_2 = w_2$.
If $w_1 = y$, then $\tilde w_1 = x$ and, by construction of $S'$,
we have $(x, w_2) = (\tilde{w}_1, \tilde{w}_2) \in E(\hs)$.
In this case, $(\tilde{w}_1, \tilde{w}_2)$ is an A2-edge in $E(\ats{T}{\hs})$. 
Otherwise, if $w_1\neq y$, then $\tilde w_1=w_1$. 
Hence, $(w_1, w_2) = (\tilde{w}_1, \tilde{w}_2) \in E(\hs)$
which implies that $(\tilde{w}_1, \tilde{w}_2)$ is an A2-edge in $E(\ats{T}{\hs})$.

\item[\textnormal{\em Case: $(w_1, w_2)$ is an A4-edge.} ] \ \\
Since $(w_1, w_2)$ is an A4-edge, there is an edge $(u,v)\in \tredge_T$
such that $w_1 = \lca_{S'}(\hmu_{S'}(u),\hmu_{S'}(v))$ and $w_2=u$. Clearly,  
$w_2=y$ is not possible, since $w_1$ corresponds to a vertex in $T$.
By construction, $\tilde w_2=w_2=u$.
Note that in $\ats{T}{\hs}$, $(u, v)$ defines the A4 edge
$(\lca_{\hs}(\hmu_{\hs}(u),\hmu_{\hs}(v)), u)$.
Therefore, it remains to show that $\tilde w_1 = \lca_{\hs}(\hmu_{\hs}(u),\hmu_{\hs}(v))$.

Notice that 
\begin{align*}
w_1 &= \lca_{S'}(\hmu_{S'}(u),\hmu_{S'}(v)) \\
    &= \lca_{S'}(lca_{S'}(\sT(u)), lca_{S'}(\sT(v)) \\
    &= \lca_{S'}(\sT(u) \cup \sT(v))
\end{align*}
In a similar manner, we obtain
\begin{align*}
\lca_{\hs}(\hmu_{\hs}(u), \hmu_{\hs}(v)) = \lca_{\hs}(\sT(u) \cup \sT(v))
\end{align*}
Let $Z = \sT(u) \cup \sT(v)$. Property (P1) implies that if $w_1 \neq y$, then 
$\lca_{\hs}(\hmu_{\hs}(u),\hmu_{\hs}(v)) = \lca_{\hs}(Z) = \lca_{S'}(Z) = w_1 = \tilde{w}_1$, as desired.
If $w_1 = y$, then $\lca_{\hs}(\hmu_{\hs}(u),\hmu_{\hs}(v)) = \lca_{\hs}(Z) = x$ and $\tilde{w}_1 = x$, as desired.
\end{description}

We have therefore shown that a cycle in $\ats{T}{S'}$ implies a cycle in
$\ats{T}{\hs}$. Since $\hs$ is a solution, we deduce that $\ats{T}{S'}$ cannot
have a cycle, and it is therefore also a solution to $((T;t,\sigma), S)$. 

To finish the
proof, we need to show that $\ats{T}{S^*}$ is acyclic. This is now easy to
see since $\hs$ can be transformed into $S^*$ by a sequence of extensions. As we
showed, each extension maintains the acyclicity property, and we deduce that
$\ats{T}{S^*}$ is acyclic.
\end{proof}

This shows that we can restrict our search to binary species trees, and we may only require that it \emph{agrees} with $\RT (T;t,\sigma)$.

\begin{proposition}
An instance $((T;t,\sigma), S)$ of the GTC problem admits a solution if and only if there exists a binary refinement $S^*$ of $S$ that 
agrees with and, therefore, displays $\RT (T;t,\sigma)$ such that $\ats{T}{S^*}$ is acyclic.  
\label{prop:IFFbinRef}
\end{proposition}
\begin{proof}
Assume that $((T;t,\sigma), S)$ admits a solution $\hs$ and let $R=\RT (T;t,\sigma)$.  
By Lemma~\ref{lem:binarize}, any binary refinement $S^*$ of $\hs$ displays $R$ (and hence agrees with it) and $\ats{T}{S^*}$ is acyclic.

Conversely, suppose that there is a binary species tree $S^*$ that is a refinement of $S$ and agrees with $R$ such that $\ats{T}{S^*}$ is acyclic.
Since $\ats{T}{S^*}$ is acyclic, we only need to show that $S^*$ displays $R$.  Let $ab|c \in R$.  Because $S^*$ is binary, we must have one of $ab|c$, $ac|b$ or $bc|a$ in $rt(S^*)$.  Since $S^*$ agrees with $R$, $ab|c \in rt(S^*)$, and it follows that $R\subseteq rt(S^*)$. Hence, $S^*$ displays $R$. Taking the latter arguments together,  $S^*$ is a solution to the instance $((T;t,\sigma), S)$ of the GTC problem, which completes the proof.
\end{proof}

\section{An Algorithm for the GTC Problem}
\label{sec:algoGTC}

We need to introduce a few more concepts before describing our algorithm. 
For a sequence $Q=(v_1,\ldots,v_k)$ we denote with $\mathcal{M}(Q) = \{v_1,\ldots,v_k\}$. 

Given a graph $G$, a 
\emph{partial topological sort} of $G$ is a sequence of distinct vertices $Q = (v_1, v_2, \ldots, v_k)$ such that for each $i \in [k]$, 
vertex $v_i$ has in-degree $0$ in $G - \{v_1, \ldots, v_{i-1}\}$.
If, for any $v \in V(G)$, there is no partial topological sort $Q'$ satisfying $\M(Q') = \M(Q) \cup \{v\}$
then $Q$ is called a \emph{maximal topological sort}. 

We note that in fact, the set of vertices in a maximal topological sort of $G$ is unique, in the sense
that for two distinct maximal topological sorts $Q,Q'$ of $G$ we always have 
$\M(Q) = \M(Q')$. 

\begin{lemma}[Properties of maximal topological sort]
	Let $G = (V,E)$ be a graph and $Q$ and $Q'$ be maximal topological sorts of $G$. 
	Then, $\M(Q) = \M(Q')$. 
	In particular, $\M(Q) = V(G)$ if and only if $G$ is a directed acyclic graph. 

	If $x\in V\setminus \M(Q)$, then none of the vertices $y$ in $V$ for 
	which there is a directed path from $x$ to $y$ are contained in $\M(Q)$.

If $x\in \M(Q)$, then $x$ is not contained in any cycle of $G$.

	\label{lem:Qproperty}
\end{lemma}
\begin{proof}
Let $Q, Q'$ be maximal topological sorts of $G$, with $Q = (v_1, \ldots, v_k)$
and assume, for contradiction that  $\M(Q) \neq \M(Q')$.
Let $v_i$ be the first vertex in the sequence $Q$ such that $v_i \notin \M(Q')$. 
Then all the in-neighbors of $v_i$ are in the set $\{v_1, \ldots, v_{i-1}\}$.
Moreover, by assumption $\{v_1, \ldots, v_{i-1}\} \subseteq \M(Q')$, implying that $v_i$ has in-degree $0$ in $G - \M(Q')$.  Hence, we could append $v_i$ to $Q'$, contradicting its maximality.  The fact that $\M(Q) = V(G)$ if and only if $G$ a directed acyclic graph is well-known and follows from the results of Kahn~\cite{Kahn:62}.

Let $x\in V\setminus \mathcal{M}(Q)$. Moreover, let 
		$P=(x=v_1,\dots  v_k=y)$, $k\geq 2$, be a directed path from $x$ to $y$.
		Since $x\notin \mathcal{M}(Q)$, $v_2$ has in-degree  greater than $0$ in $G-\M(Q)$. 
		Therefore, $v_2\notin \M(Q)$ and, by induction,
	  $v_k=y\notin \M(Q)$.

We now show that no vertex $x\in \M(Q)$ can be contained in a cycle of $G$.
		Assume, for contradiction, that there is a cycle $C$ such that some of its vertices are part of 
		a maximal topological sort $Q = (v_1,\dots,v_k)$ of $G$. 
		Let $v_i$ be the first vertex of $C$ that appears in $Q$. 
		Hence, $v_i$ must have  in-degree $0$ in 		$G - \{v_1, \ldots, v_{i-1}\}$.
		But this implies, that the in-neighbor of $v_i$ in $C$
		must already be contained in $Q$; a contradiction.
\end{proof}

A maximal topological sort of $G$ can be found by
applying the following procedure: start with $Q$ empty, and while there is a
vertex of in-degree $0$ in $G - \M(Q)$, append $v$ to $Q$ and repeat. Then, $G$ is acyclic if an only if any maximal 
topological sort $Q$ of $V(G)$ satisfies $\M(Q) = V(G)$. The latter argument is correct as it directly 
mirrors the well-known algorithm by Kahn to find a topological sort of graph \cite{Kahn:62}.

Our algorithm will make use of what we call a \emph{good split refinement}.
To this end, we provide first
\begin{definition}[Split refinement]
\label{def:split-ref}
Let $S$ be an almost binary tree and let $x$ be a cherry of $S$.
We say that a refinement $S'$ of $S$ is a \emph{split refinement (of $S$ at $x$)} 
if $S'$ can be obtained from $S$ by 
partitioning the set $\child(x)$ of children of $x$
into two non-empty subsets $X_1, X_2=\child(x)\setminus X_1$,
and applying the extensions $(x, X_1)$ and then $(x, X_2)$. 
\end{definition}

In other words, we split the children set of $x$ into
non-empty subsets $X_1$ and $X_2$, and add a new parent vertex above each
subset of size 2 or more and connect $x$ with the newly created parent(s)
or directly with $x'$ whenever $X_i=\{x'\}$.

We note that the two $(x, X_1)$ and $(x, X_2)$ extensions yield a valid refinement of $S$
since the set $X_2$ is a strict subset of the children of $x$ in $S_{x, X_1}$. 
Also observe that a split refinement transforms
an almost binary tree into another almost binary tree that has one additional
binary internal vertex.

\begin{definition}[Good split refinement]
Let $((T;t,\sigma), S)$ be a GTC instance.  Let $Q$ be a maximal topological sort of $\ats{T}{S}$, and let $S'$ be a split refinement of $S$ at some cherry vertex $x$. 
Then $S'$ is a \emph{good split refinement} if the two following conditions are satisfied:
\begin{itemize}
    \item 
    $S'$ agrees with $\RT (T;t,\sigma)$;
    
    \item
    all the in-neighbors of $x$ in $\ats{T}{S'}$ belong to $\M(Q)$.
\end{itemize}
\end{definition}

\begin{figure}[tbp]
	\begin{center}
		\includegraphics[width=.9\textwidth]{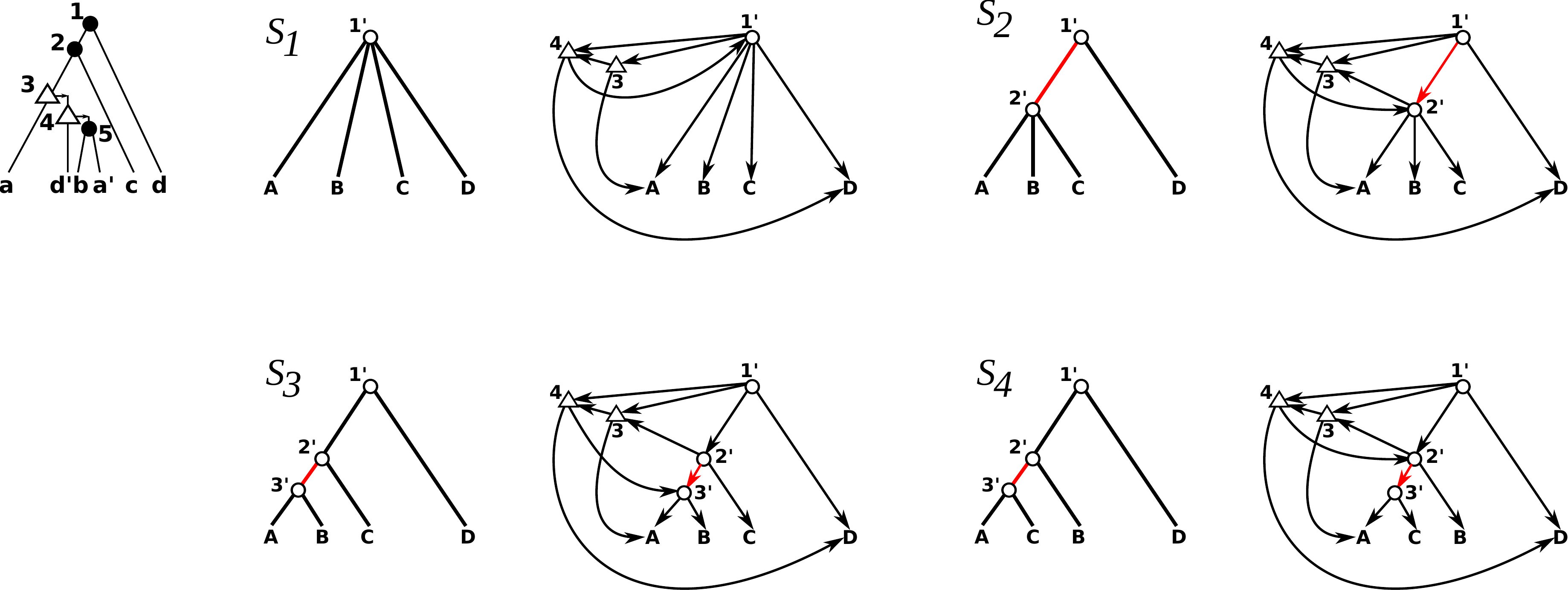}
	\end{center}
	\caption{Top left: the gene tree $(T;t,\sigma)$  from Fig.\ \ref{fig:least}	
		from which we obtain the species triples $\S(T;t,\sigma) = \{AB|D,AC|D\}$.
		Moreover, the sequence of species trees $S_1, S_2$ and $S_3$ is
		obtained by stepwise application of good split refinements.  The species tree $S_4$ is an example of a split refinement of $S_2$ that is not good.
		The corresponding graphs  $\ats{T}{S}$ are drawn right to the respective species tree $S$. 
		For clarity, we have omitted to draw all vertices of $\ats{T}{S}$ that have degree $0$.  
		See text for further discussion.
	}
\label{fig:working-exmpl}
\end{figure}

The intuition behind a good split refinement is that it refines $S$ by creating an additional binary vertex.  Moreover, this refinement maintains agreement with $\RT (T;t,\sigma)$ and, more importantly, creates a new vertex of in-degree $0$ in the auxiliary graph that can be used to extend the current maximal topological sort.  Ultimately, our goal is to repeat this procedure until $Q$ contains every vertex, at which point we will have attained an acyclic graph.
As an example consider Fig.\ \ref{fig:working-exmpl}. The species tree $S_1$ corresponds to the star tree.
	Clearly $S_1$ agrees with $\S(T;t,\sigma)$ since $R(S_1)=\emptyset$. However,  $\ats{T}{S}$ contains cycles.
	For the maximal topological sort $Q_1$ of  $\ats{T}{S_1}$ we have  
	$\M(Q_1) = L(T)\cup \{1,2,5\}$. 
	Now, $S_2$ is a good split refinement of $S_1$, since $S_2$ agrees with $\S(T;t,\sigma)$ (in fact, $S_2$ displays  $\S(T;t,\sigma)$)
	and since $x=1'$ has no in-neighbors in $\ats{T}{S_2}$ which trivially implies that 
	all in neighbors of $x=1'$ in $\ats{T}{S_2}$ are already contained $\M(Q_1)$.
	For the maximal topological sort $Q_2$ of  $\ats{T}{S_2}$ we have  
	$\M(Q_2) = \M(Q_1)\cup \{1'\}$.  
	Still,  $\ats{T}{S_2}$ is not acyclic. 
	The tree $S_3$ is a good split refinement of $S_2$, since $S_3$ agrees with $\S(T;t,\sigma)$	
	and the unique in-neighbor $1'$ of $x=2'$ in $\ats{T}{S_3}$ is already contained $\M(Q_2)$.
	Since  $\ats{T}{S_3}$ is acyclic, there is a time-consistent reconciliation map from  
	$(T;t,\sigma)$ to $S_3$, which is shown in 		Fig.\ \ref{fig:least}.
	Furthermore, $S_4$ is not a good split refinement of $S_2$. 
	Although $S_4$ is a split refinement of $S_2$ and agrees with $\S(T;t,\sigma)$,
	the in-neighbor $4$ of $x=2'$ is not contained in $\M(Q_2)$.

We will discuss later the question of finding a good split refinement efficiently, if one exists.  For now, assume that this can be done in polynomial time.
The pseudocode for a high-level algorithm for solving the GTC problem is provided in Alg.\ \ref{alg:gtcRefinement}. We note in passing that this
algorithm serves mainly as a scaffold to provide the correctness proofs that are needed for the main Alg.\ \ref{alg:GoodSplit}.

\vspace{3mm}
\begin{algorithm}[H]
\begin{algorithmic}[1]
\State Function $\textit{gtcRefinement}((T;t,\sigma), S)$

    \If{$S$ is binary \label{line:binary}} 
        \If{$\ats{T}{S}$ is acyclic and $S$ agrees with $\RT (T;t,\sigma)$ \label{line:GTC-properties}} %
           \State return $S$\;
        \Else \ return ``there is no solution''\;
        \EndIf
    \Else
        \If{$S$ admits a good split refinement $S'$ at a vertex $x$ \label{line:good-split}} 
           \State return $\textit{gtcRefinement}((T;t,\sigma), S')$\;
        \Else           \ return ``there is no solution''\; \label{line:no-good-split} 
        \EndIf
    \EndIf
 \end{algorithmic}
 \caption{GTC algorithm}\label{alg:gstr}
	\label{alg:gtcRefinement}
\end{algorithm}

We prove some general-purpose statements first.
Let $I_G(v)$ denote the set of in-neighbors of vertex $v$ in a graph $G$.

\begin{lemma}\label{lem:ancestors-dont-change}
Let $((T;t,\sigma), S)$ be a GTC instance. Moreover, 
let $S'$ be a split refinement of $S$ at a cherry $x$.  
Then, for every vertex $y$ of $S$ such that $y \not\preceq_S x$, it holds that 
$I_{\ats{T}{S'}}(y) = I_{\ats{T}{S}}(y)$.
\end{lemma}

\begin{proof}
	Let $y\in V(S)$ be a vertex of $S$ satisfying $y \not\preceq_S x$.
		Since $y\notin V(T)$, for every $z\in I_{\ats{T}{S'}}(y)$ or $z\in I_{\ats{T}{S}}(y)$ the edge $(z,y)$
		cannot be an A4-edge. 	

		If $(z,y)$ is an A2-edge in $\ats{T}{S}$ then, $(z,y)\in E(S)$, which is if and only if 
		$(z,y)\in E(S')$, since  $y\not\preceq_S x$. In this case, 
		$z\in I_{\ats{T}{S}}(y)$ if and only if $z\in I_{\ats{T}{S'}}(y)$, 

		It remains to consider A1- and A3-edges. We translate here Property (P2.a) as in the proof of Lemma \ref{lem:binarize}. 
		It states that $\hmu_{S}(y)=\hmu_{S'}(y)$ since  $y\not\preceq_{S'} x$ and thus, it cannot be the newly created vertex in $S'$. 
		Since this holds for every $y \not\preceq_S x$, this immediately implies that 		
		$z\in I_{\ats{T}{S}}(y)$ and the edge $(z,y)$ is an A1-edge, resp., A3-edge in $\ats{T}{S}$
		if and only if $z\in I_{\ats{T}{S'}}(y)$ and $(z,y)$ is an A1-edge, resp., A3-edge in $\ats{T}{S'}$.
\end{proof}

\begin{remark}
Lemma \ref{lem:self-loops} implies that if $S$ has a leaf that is in a self-loop in $\ats{T}{S}$, then we may immediately discard $S$ as it cannot have a solution (since any refinement will have this self-loop). For the rest of the section, we will therefore assume that no leaf of $S$ belong to a self-loop.
\label{rem:self-loop}
\end{remark}

We now show that if we reach a situation where there is no good split refinement, then there is no point in continuing, i.e. that it is correct to deduce that no solution refining the current $S$ exists.

\begin{proposition}
Let $((T;t,\sigma), S)$ be a GTC instance such that $S$ is not binary   and
does not admit a good split refinement. 
Then,  $((T;t,\sigma), S)$ does not admit a solution. 
\label{prop:no-solution}
\end{proposition}

\begin{proof}
We show that if $S$  does not admit a good split refinement, then 
none of the binary refinements $S^*$ of $S$ is a solution to the 
GTC instance $((T;t,\sigma), S)$. 
Contraposition of Lemma \ref{lem:binarize} together with Prop.\ \ref{prop:IFFbinRef} then implies 
that there is no solution at all for $((T;t,\sigma), S)$.

Thus, assume that $S$ is not binary (but almost binary, due to the definition of GTC instances)  
such that $S$ does not admit a good split refinement. 
Let $S^*$ be any binary refinement of $S$.  We may assume that $S^*$ agrees with and thus,  displays  $\RT (T,t,\sigma)$, as otherwise it is not a solution. 
We show that $\ats{T}{S^*}$ contains a cycle.

Let $Q$ be a maximal topological sort of $\ats{T}{S}$. By Lemma \ref{lem:Qproperty}, $\M(Q)$ is independent of the choice of 
the particular sequence $Q$. 
Note that $V(S) \subseteq V(S^*)$ and therefore that $V(\ats{T}{S}) \subseteq V(\ats{T}{S^*})$. 
In particular, $\M(Q) \subseteq V(\ats{T}{S^*})$.
Also notice that because of the A2 edges in $\ats{T}{S}$, if a vertex $x \in V(S)$ is not in $\M(Q)$, then no descendant of $x$ in $S$ is in $\M(Q)$.
We separate the proof into three claims.

\begin{owndesc}
\item[Claim 1:]
\emph{Let $x$ be a non-binary cherry of $S$.  Then $x \notin \M(Q)$.}

Note that since
$x$ is non-binary and  $S^*$ is
a binary refinement of $S$, there is a split refinement 
$S'$ of $S$ at $x$ such
that $S^*$ refines $S'$.  
Since $S^*$ agrees with $\RT (T;t,\sigma)$, also $S'$ agrees with $\RT (T;t,\sigma)$.  
If all in-neighbors of $x$ in $\ats{T}{S'}$ are in $Q$, then $S'$ is a  good split refinement; a contradiction. 
So we may assume that $x$ has an in-neighbor $y$ in $\ats{T}{S'}$ such that $y \notin \M(Q)$.  
Since $x\in V(S)$, the edge $(y, x)$ cannot be an A4-edge in $\ats{T}{S'}$.
If $(y, x)$ is an A1-edge in $\ats{T}{S'}$, then $x = \hmu_{S'}(v)$ for some $v \in V(T)$.  
By construction, $L(S(x)) = L(S'(x))$ and thus, $\hmu_{S}(v) = \hmu_{S'}(v) = x$. Therefore,
 $(y, \hmu_S(v)) = (y, x) \in E(\ats{T}{S})$.
Similarly, if $(y, x)$ is an A3-edge in $\ats{T}{S'}$, then $x = \hmu_{S'}(y)$ and again, 
$(y, \hmu_S(y)) = (y, x) \in E(\ats{T}{S})$.  
If $(y, x)$ is an A2-edge in $\ats{T}{S'}$, then $(y, x) \in E(\ats{T}{S})$ since the parent
of $x$ is the same in $S$ and $S'$. In all cases, $y$ is an in-neighbor of
$x$ in $\ats{T}{S}$. However, since $y \notin \M(Q)$, vertex $y$ remains an in-neighbor of $x$ in 
the graph $\ats{T}{S} - \M(Q)$.  It follows that $x \notin \M(Q)$, which proves Claim 1.
\end{owndesc}

\begin{owndesc}
\item[Claim 2:]
\emph{Let $v \in V(T) \setminus \M(Q)$.  Then $v$ has in-degree at least $1$ in $\ats{T}{S^*} -  \M(Q)$.}

Let $v \in V(T) \setminus \M(Q)$.  Since $v \notin \M(Q)$, $v$ has in-degree at least $1$ in $\ats{T}{S} - \M(Q)$, or else it could be added to the maximal topological sort.
Let $(x, v)$ be an incoming edge of $v$ in $\ats{T}{S} - \M(Q)$, which is either an A1- or an A4- edge.  

If $(x, v)$ is an A1- edge, we either have $x \in V(T)$ or $x \in V(S)$.
Suppose first that $x \in V(T)$. In this
case, the $(x, v)$ edge exists because $x$ is the parent of $v$ in $T$
with $t(x), t(v)$ both in $\{\dup, \transfer\}$. This is independent of
$S$, and so $(x, v)$ is also an A1-edge of $\ats{T}{S^*} - \M(Q)$.
Suppose now that $x \in V(S)$. In this case, observe that
$x\notin \M(Q)$, since $(x, v)$ is an edge in $\ats{T}{S} - \M(Q)$.
Therefore,  $x \in V(S) \setminus \M(Q)$. This, in particular, implies
that the parent $v_p$ of $v$ in $T$ satisfies $t(v_p) = \spec$ and
$\hmu_{S}(v_p) = x$. Since $S^*$ refines $S$, we must have $\hmu_{S^*}(v_p)
\preceq_{S^*} x$.  
There are two cases, either
$\hmu_{S^*}(v_p)\notin V(S)$, in which case trivially $\hmu_{S^*}(v_p) \notin
\M(Q)$, or $\hmu_{S^*}(v_p)\in V(S)$. In the latter case, there is a
directed (possibly edge-less) path from $x$ to $\hmu_{S^*}(v_p)$ in $\ats{T}{S}$ due to the A2-edges. 
Thus, we can apply  Lemma
\ref{lem:Qproperty} to conclude that $\hmu_{S^*}(v_p) \notin \M(Q)$. 

In either
case, $(\hmu_{S^*}(v_p), v)$ is an A1-edge of $\ats{T}{S^*} - \M(Q)$.
Therefore, $v$ has an in-neighbor in $\ats{T}{S^*}$ that does not belong to
$Q$.

Assume now that $(x, v)$ is an A4-edge.
Thus, there is an edge $(v,v') \in \tredge_T$ such that 
$x= lca_{S}(\hmu_S(v), \hmu_S(v'))$. 
Again since $S^*$ refines $S$, it is not hard to see that 
$\lca_{S^*}(\hmu_S^*(v), \hmu_S^*(v')) \preceq_{S^*} x$. 
By similar arguments as before,  $\lca_{S^*}(\hmu_S^*(v), \hmu_S^*(v')) \notin \M(Q)$.  
Thus, $(\lca_{S^*}(\hmu_S^*(v), \hmu_S^*(v')),v)$ is an A4-edge of
 $\ats{T}{S^*}$. Hence, 
$v$ also has in-degree at least $1$ in $\ats{T}{S^*} - \M(Q)$, which proves Claim 2.
\end{owndesc}

We prove the analogous statement for the species tree vertices.

\begin{owndesc}
\item[Claim 3:]
\emph{Let $x \in V(S) \setminus \M(Q)$.  Then $x$ has in-degree at least $1$ in $\ats{T}{S^*} - \M(Q)$.}

Let $x \in V(S) \setminus \M(Q)$.
We may assume that $x$ has in-degree at least $1$ in $\ats{T}{S} - \M(Q)$, by the maximality of $Q$. 
Notice that since $S^*$ is a binary refinement of $S$, there exists a sequence of split refinements that transforms $S$ into $S^*$.  That is, 
there is a sequence of trees $S = S_1, S_2, \ldots, S_k = S^*$ such that for $2 \leq i \leq k$, $S_i$ is a split refinement of $S_{i-1}$.

Let $(w, x)$ be an incoming edge  of $x$ in $\ats{T}{S} - \M(Q)$. We consider the following three exclusive cases: 
either $x$ is a binary or a non-binary interior vertex, or a leaf.

Suppose first that $x$ is a binary vertex of $S$.  Because $S$ is almost binary, $x$ is not a descendant of any non-binary vertex of $S$.  
By applying Lemma~\ref{lem:ancestors-dont-change} successively on each split refinement of the sequence transforming $S$ into $S^*$, 
we obtain $I_{\ats{T}{S_1}}(x) = I_{\ats{T}{S_2}}(x) = \ldots = I_{\ats{T}{S_k}}(x) = I_{\ats{T}{S^*}}(x)$.  In particular, $w \in I_{\ats{T}{S^*}}(x)$, which proves the claim for this case since $w \notin \M(Q)$.

Suppose now that $x$ is a leaf of $S$. If the parent $x_p$ of $x$ is binary,
then again, successive application of Lemma~\ref{lem:ancestors-dont-change} on
$S_1, \ldots, S_k$ implies that $I_{\ats{T}{S}}(x) = I_{\ats{T}{S^*}}(x)$, and
therefore that $w \in I_{\ats{T}{S^*}}(x)$. If $x_p$ is a non-binary cherry,
then $x_p \notin \M(Q)$ by Claim 1. There are two cases, either the parent $p(x)$
of $x$ in $S^*$ is identical to $x_p$ or not.
In the first case, $p(x)=x_p$ is not part of $Q$. 
In the latter case, $p(x)$ refers to some newly added vertex during 
the construction of $S^*$. In this case, $p(x)$ is not contained in $S$ and so neither in $Q$. 
In summary,
the parent of $x$ in $S^*$
is not in $Q$. Due to the A2-edges, $x$ has in-degree at least $1$ in
$\ats{T}{S^*} - \M(Q)$.

Finally, suppose that $x$ is a non-binary interior vertex of $S$, i.e. $x$ is a cherry.
Let $S'$ be a split refinement of $S$ at $x$ such that $S^*$ refines $S'$.  Recall that as in Claim 1, $S'$ agrees with $\RT (T;t,\sigma)$.
This and the fact that $S$ does not admit a good split refinement implies 
 that $x$ has  in-degree at least $1$ in $\ats{T}{S'} - \M(Q)$.
Now, $x$ is binary in $S'$.  As before, there is a sequence of binary refinements transforming $S'$ into $S^*$.  Since $x$ is not a descendant of any non-binary vertex in $S'$, by applying Lemma~\ref{lem:ancestors-dont-change} on each successive refinement, 
$I_{\ats{T}{S'}}(x) = I_{\ats{T}{S^*}}(x)$.
It follows that $x$ has in-degree at least $1$ in $\ats{T}{S^*} - \M(Q)$ as well. 

This proves Claim 3.
\end{owndesc}

Now, let $y \in V(S^*) \setminus V(S)$.  Thus, $y$ must have been created by one of the extensions that transforms $S$ into $S^*$, and so in $S^*$, $y$ must be a descendant of a vertex $x$ such that $x$ is a cherry in $S$.  Since $x \notin \M(Q)$ by Claim 1, and because of the A2-edges, $y$ must have in-degree at least $1$ in $\ats{T}{S^*}  - \M(Q)$.

To finish the argument, note that $V(\ats{T}{S^*} - \M(Q)) = (V(T) \setminus \M(Q)) \cup (V(S) \setminus \M(Q)) \cup (V(S^*) \setminus V(S))$.
We just argued that each vertex in $V(S^*) \setminus V(S)$ has in-degree at least $1$ in $\ats{T}{S^*} - \M(Q)$, and by Claim 2 and Claim 3, it follows that every vertex of $\ats{T}{S^*} - \M(Q)$ has in-degree at least $1$.  This implies that $\ats{T}{S^*} - \M(Q)$ contains a cycle, and hence that $\ats{T}{S^*}$ also contains a cycle.  We have reached a contradiction, proving the lemma.
\end{proof}

We next show that if we are able to find a good split refinement $S'$ of $S$, the $((T; t, \sigma), S')$ instance is equivalent in the sense that 
$((T; t, \sigma), S)$ admits a solution if and only if $((T; t, \sigma), S')$ also admits a solution.
First, we provide the following lemma for later reference.

\begin{lemma}\label{lem:q-stays-topo}
Let $((T;t,\sigma), S)$ be a GTC instance and let $Q$ be a maximal topological sort of $\ats{T}{S}$. Moreover, 
let $S'$ be a split refinement of $S$ at a cherry $x$.
Then, for any maximal topological sort $Q'$ of $\ats{T}{S'}$, it holds that 
$\M(Q) \subseteq \M(Q')$.
\end{lemma}

\begin{proof}
Assume without loss of generality that the cherry $x$ is non-binary in $S$, as otherwise $S=S'$ and we are done. 
Let $x_1, x_2$ be the children of $x$ in $S'$, and assume furthermore w.l.o.g. 
that $|\L(S'(x_1))| \geq |\L(S'(x_2))|$.
Note that $x_2$ could be a leaf, but that $x_1$ must be an internal vertex since $x$ is a non-binary cherry.
Now, if $\M(Q)=\emptyset$, then the lemma statement is trivially satisfied. Hence, assume that 
$Q = (w_1, \ldots, w_l)$, $l\geq 1$.
We construct partial topological sorts $Q_0, Q_1, \ldots, Q_l$ of $\ats{T}{S'}$ as follows.
Define $Q_0 = ()$ as an empty sequence and, for each $1 \leq i \leq l$, 
$Q_i$ is obtained from $Q_{i-1}$ by appending $w_i$ to $Q_{i-1}$ if $w_i \neq x$, 
and if $w_i = x$, by appending $x$ and $x_1$ (in this order) to $Q_{i-1}$, 
and then appending $x_2$ to $Q_{i-1}$ if it is not a leaf in $S'$. 
We show, by induction, that each $Q_i$ is a partial topological sort of $\ats{T}{S'}$.  The base case $i = 0$ is clearly satisfied.   
So let us assume that for $i > 0$ the sequence $Q_{i - 1}$ 
is a partial topological of of $\ats{T}{S'}$.  Consider now the vertex $w_i$.

\begin{owndesc}
\item[\textnormal{\em Case: $w_i \in V(S)$ and $w_i \not\preceq_S x$} ] \ \\
By Lemma~\ref{lem:ancestors-dont-change}, $I_{\ats{T}{S}}(w_i) = I_{\ats{T}{S'}}(w_i)$.  Since each member of $I_{\ats{T}{S}}(w_i)$ precedes $w_i$ in $Q$, 
$\M(Q_{i-1})$ contains $I_{\ats{T}{S}}(w_i)$.  It follows that appending $w_i$ 
to $Q_{i-1}$ yields a partial topological sort $Q_i$ of $\ats{T}{S'}$.
\end{owndesc}

In the remaining cases, we will make frequent use of the fact that if 
$Q'$ is a partial topological sort of $\ats{T}{S'}$ and $v$ is a vertex with 
 $I_{\ats{T}{S'}}(v)\setminus M = \emptyset$ for some (possibly empty) subset $M\subseteq \M(Q')$, then appending $v$ to $Q'$
yields a partial topological sort of $\ats{T}{S'}$. 
In other words, we can w.l.o.g.\ assume $I_{\ats{T}{S'}}(v)\setminus M\neq \emptyset$
for all such considered sets. 

\begin{owndesc}
\item[\textnormal{\em Case: $w_i \in V(S)$ and $w_i = x$} ] \ \\
We start by showing that the sequence $Q^x_{i-1}$ obtained by appending $x$ to $Q_{i-1}$
is a partial topological sort of $\ats{T}{S'}$.
Let $z \in I_{\ats{T}{S'}}(x)$. 
Suppose first that $z \in V(S')$. 
Then $(z, x)$ is either an A1- or A2-edge of $\ats{T}{S'}$.  If $(z, x)$ is an A2-edge, then $z$ is the parent of $x$ in both $S$ and $S'$.  Thus $(z, x) \in E(\ats{T}{S})$ and since $x \in \M(Q)$, we must have $z \in \M(Q)$.  Moreover, $z$ must precede $x$ in $Q$, and it follows that $z \in \M(Q_{i-1})$.  If $(z, x)$ is an A1-edge, then $x \preceq_{S'} z$.  
If $x \prec_{S'} z$, then $x \prec_{S} z$ as well.  Thus in $\ats{T}{S}$, there is a path of A2-edges from $z$ to $x$, implying that $z$ precedes $x$ in $Q$.
Finally if $x = z$, then $\ats{T}{S'}$ contains the self-loop $(x, x)$.  
In this case, there is an edge $(u, v) \in E(T)$ such that $(x, x) = (\hmu_{S'}(u), \hmu_{S'}(v))$.  By construction, $\L(S'(x)) = \L(S(x))$ and therefore, $(\hmu_S(u), \hmu_S(v)) = (x, x)$ is an edge of $\ats{T}{S}$.  This case cannot occur, since it is impossible that $x \in \M(Q)$ if $x$ is part of a self-loop.
Therefore, $z$ precedes $x$ in $Q$ whenever $z \in V(S')$.

If instead $z \in V(T)$, then $(z, x)$ is either an A1- or A3-edge of $\ats{T}{S'}$, 
in which case there is $z'\in V(T)$ such that $(z, x) = (z, \hmu_{S'}(z'))$. 
By construction  $\L(S'(x)) =\L(S(x))$ and therefore, $(z, \hmu_{S}(z')) = (z, x)$ is an edge of $\ats{T}{S}$. 
Again, $z$ must precede $x$ in $Q$.  We have thus shown that $z$ precedes $x$ in $Q$ for every $z \in I_{\ats{T}{S'}}(x) \subseteq \M(Q_{i-1})$.
Hence, appending $x$ to  $Q_{i-1}$ yields a partial topological sort $Q^x_{i-1}$ of $\ats{T}{S'}$.

We continue with showing that $Q_{i}$ is a partial topological sort of $\ats{T}{S'}$.
Note,  $Q_{i}$ is obtained by appending $x_1$ and, in case $x_2$ is not a leaf in $S'$, also $x_2$ to the partial topological sort  $Q^x_{i-1}$ of $\ats{T}{S'}$.
Let $z \in I_{\ats{T}{S'}}(x_j)\setminus \{x\}$, where $x_j \in \{x_1, x_2\}$ is is chosen to be an interior vertex of $S'$. 
Note, $x_j=x_1$ is always possible as argued at the beginning of this proof.
Suppose that $z \in V(S')$. In this case,  $(z, x_j)$ cannot be an A2-edge since it would imply $x=z$; a contradiction.
Hence, $(z, x_j)$ is an A1-edge of $\ats{T}{S'}$ and $x_j \preceq_{S'} z$.
Similarly as before, if $x_j \prec_{S'} z$, then $x \prec_{S} z$ since $z
\neq x$. Thus, $z$ precedes $x$ in $Q$, since $\ats{T}{S}$ contains a path of A2-edges from $z$ to
$x$. If $x_j = z$, then there is an edge $(u, v) \in E(T)$
such that $(\hmu_{S'}(u), \hmu_{S'}(v)) = (x_j, x_j)$. Since $x_j$ is supposed not to be a
leaf in $S'$ and by construction of $S'$ from $S$,  we must have in $S$ that $(\hmu_S(u), \hmu_S(v)) = (x, x)$, contradicting $x
\in \M(Q)$. Now, assume that $z \in V(T)$ in which case $(z, x_j)$ is
either an A1- or A3-edge in $\ats{T}{S'}$. Again, there must be a vertex
$z'\in V(T)$ such that $(z, x_j) = (z, \hmu_{S'}(z'))$. By construction
$\L(S'(x_i)) \subset \L(S(x))$. This and $x_j = \hmu_{S'}(z')$ immediately
implies that $x = \hmu_{S}(z')$. 
Thus, $(z, \hmu_{S}(z')) = (z, x)$ is an edge of $\ats{T}{S}$ and  $z$ must precede $x$ in $Q$. 
Again, 
this holds for every $z$ which implies implies that $I_{\ats{T}{S'}}(x_j) \setminus \{x\} \subseteq \M(Q^x_{i-1})$.
Thus, appending $x_1$ and $x_2$ to $Q^x_{i-1}$ after $x$ yields the partial topological sort $Q_i$ of $\ats{T}{S'}$.

\item[\textnormal{\em Case: $w_i \in V(S)$ and $w_i \prec_S x$} ] \ \\
Since $x$ is a cherry, $w_i$ must be a leaf in $S$.
Thus $x$ precedes $w_i$ in $Q$ and therefore, we may assume that $x, x_1$
and, in case $x_2$ is not a leaf in $S'$, also $x_2$ are contained in the partial topological sort $Q_{i-1}$
of $\ats{T}{S'}$ Note, $x_2$ could be absent from $Q_{i-1}$ if it is a leaf and $w_i = x_2$. 
That is, we may assume that $w_i$ is a child of either $x, x_1$ or $x_2$ in $S'$, 
and that the parent of $w_i$ in $S'$ is in $Q_{i-1}$.
Consider $z \in I_{\ats{T}{S'}}(w_i) \setminus \{x,
x_1, x_2\}$. If $z \in V(S')$, then $w_i \preceq_{S'} z$.  As before, if $w_i \prec_{S'} z$, then $w_i \prec_{S} z$ and because of the A2-edges of $\ats{T}{S}$, 
$z$ precedes $w_i$ in $Q$.   If $z = w_i$, then $(w_i, w_i)$ is a self-loop in $\ats{T}{S'}$.  
By Lemma~\ref{lem:self-loops}, 
$(w_i, w_i)$ is also a self-loop of $\ats{T}{S}$; a contradiction since, by Remark \ref{rem:self-loop}, 
$\ats{T}{S}$ has no self-loops on its leaves.
So assume that $z \in V(T)$. Then $(z, w_i)$ is an
A1- or A3-edge. Since $w_i$ is a leaf, we have that for any $v \in V(T)$,
$\hmu_{S'}(v) = w_i$ if and only if $\hmu_S(v) = w_i$. It follows that $(z,
w_i) \in E(\ats{T}{S})$. Therefore, $z$ precedes $w_i$ in $Q$ and $z$
belongs to $Q_{i-1}$. Thus we may append $w_i$ to $Q_{i-1}$ to obtain a
partial topological sort $Q_i$ of $\ats{T}{S'}$.

\item[\textnormal{\em Case: $w_i \in V(T)$} ] \ \\
Let $z \in I_{\ats{T}{S'}}(w_i)$. Thus, $(z,w_i)$ is either an A1- or A4-edge in $\ats{T}{S'}$. 
If $z \in V(T)$, then $(z,w_i)$ is an A1-edge in  $\ats{T}{S'}$. Since the event-labels in $T$ are fixed, 
$(z,w_i)$ is an A1-edge in  $\ats{T}{S}$ and thus, 
 $z \in I_{\ats{T}{S}}(w_i)$. Therefore, $z$ precedes $w_i$ in $Q$. 

Now, suppose $z \in V(S')$. 
If $(z, w_i)$ is an A1-edge, then the parent $u$ of $w_i$ in $T$ satisfies
$\hmu_{S'}(u) = z$. If $z \notin \{x, x_1, x_2\}$, then we immediately obtain $\hmu_{S}(u) = z$.
Hence, $(z, w_i) \in E(\ats{T}{S})$ and thus, $z$ precedes
$w_i$ in $Q$. If $z \in \{x, x_1, x_2\}$, then it is easy to verify that $\hmu_S(u) =
x$. Thus $(x, w_i) \in E(\ats{T}{S})$, $x$ precedes $w_i$ in $Q$. 
By construction, we have added $x,x_1,x_2$ in one of the previous steps
to obtain $Q_j$, $1\leq j\leq i-1$. Hence,  $z \in \{x, x_1, x_2\} $ precedes $w_i$ in $Q$.

If instead $(z, w_i)$ is an A4-edge, then $w_i$ has a child $v$ such that 
$\lca_{S'}(\hmu_{S'}(w_i), \hmu_{S'}(v)) = z$.
Clearly, it holds that   $z = \lca_{S'}(Z)$ for 
$Z = \sT(w_i) \cup \sT(v)$.
If $z \in \{x, x_1, x_2\}$, then $(\lca_S(Z), w_i) = (x, v) \in E(\ats{T}{S})$, 
and if $z \notin \{x, x_1, x_2\}$, then $(\lca_S(Z), w_i) = (\lca_{S'}(Z), w_i) = (z, w_i) \in E(\ats{T}{S})$. 
In both cases, $z$ precedes $w_i$ in $Q$.

In every case, each $z$ is already contained in $Q_{i-1}$, and we may append $w_i$ to $Q_{i-1}$
to obtain a partial topological sort $Q_i$ of $\ats{T}{S'}$.

\end{owndesc}

We have shown that $Q_l$ is a partial topological sort of $\ats{T}{S'}$
satisfying $\M(Q) \subseteq \M(Q_l)$. If we add in-degree $0$ vertices in
$Q_l$ until we obtain a maximal topological sort $Q'$ of $\ats{T}{S'}$, then 
have $\M(Q) \subseteq \M(Q_l) \subseteq \M(Q')$, as desired.
\end{proof}

Our last step is to show that any good split refinement leads to a solution, if any.

\begin{theorem}\label{thm:equiv-refinement}
Let $((T; t, \sigma), S)$ be a GTC instance, and suppose that $S$ admits a
good split refinement $S'$. Then $((T;t, \sigma), S)$ admits a solution if and
only if $((T; t, \sigma), S')$ admits a solution.
Moreover, any solution for $((T;t, \sigma), S')$, if any, is also a solution for $((T;t,\sigma), S)$.
\end{theorem}

\begin{proof}
It is easy to see that any solution $S^*$
of $((T; t, \sigma), S')$ would be a solution for  $((T; t, \sigma), S)$.
Hence, if $((T; t, \sigma), S)$ does not admit a solution,
then $((T; t, \sigma), S')$ cannot admit a solution.

\begin{figure}[h]
	\begin{center}
	\includegraphics[width=0.9\textwidth]{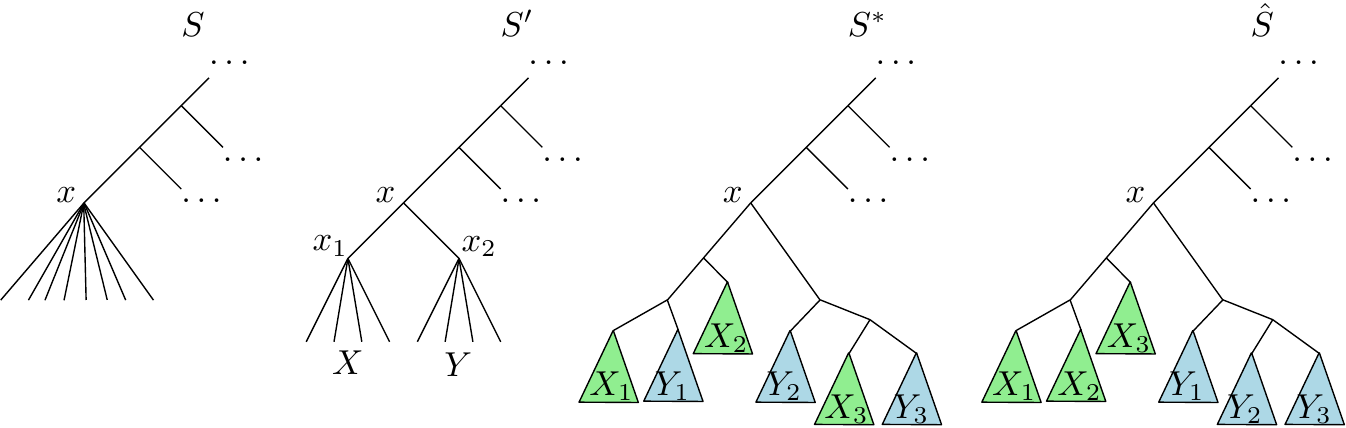}
	\end{center}
	\caption{A representation of the trees $S, S', S^*$ and $\hs$.  The $X_1, X_2$ and $X_3$ triangles represent subtrees containing only leaves from $X$ (same with $Y_1, Y_2, Y_3$ and $Y$).  }
	\label{fig:all-the-s-guys}
\end{figure}

For the converse, suppose now that $((T; t, \sigma), S)$ admits a solution. Thus, 
there is a binary refinement $S^*$ of $S$
that displays $\RT (T;t,\sigma)$ and such that $\ats{T}{S^*}$ is acyclic.
Let $S'$ be any good split refinement of $S$ at some cherry $x$ of $S$. 
Furthermore,
let $x_1, x_2$ be the children of $x$ in $S'$, and let $X = \L(S'(x_1))$
and $Y = \L(S'(x_2))$. Note that $\{X,Y\}$ is a partition of the children
of $x$ in $S$.
Consider the trees $S^*|_X$ and $S^*|_Y$. We define another tree $\hs$ obtained by replacing the children of $x$ in $S^*$ by $S^*|_X$ and $S^*|_Y$.  More precisely, first observe that, by construction 
$\L(S^*|_X)= X$ and $\L(S^*|_Y)= Y$. Moreover, for any binary refinement $S^*$
of $S$ it must hold that $\L(S^*(x))$ is the set of children $\child(x)$ in $S$. In particular, $x$ is an 
ancestor in $S^*$ of every vertex in $S^*|_X$ as well as in $S^*|_Y$. 
Hence, we can safely replace the two subtrees $S^*(v_1)$ and $S^*(v_2)$ rooted  at the two children $v_1,v_2$ of $x$ in $S^*$ by 
$S^*|_X$ and $S^*|_Y$ (by defining the root of $S^*|_X$ and the root of $S^*|_Y$ as the two 
new children of $x$) to obtain another tree $\hs$ with  $\L(\hs) = \L(S^*)$. 
By construction, $\hs$ is identical to $S^*$, except that the two subtrees below $x$ are replaced by $S^*|_X$ and $S^*|_Y$.
An example of the trees $S, S', S^*$ and $\hs$ is shown in Figure~\ref{fig:all-the-s-guys}.

Clearly, $\hs(x_1)= S^*|_X$, resp., $\hs(x_2) =S^*|_Y$ is a binary refinement of 
$S'(x_1) $, resp., $S'(x_2)$. Moreover, $S^*|_{\L(S^*)\setminus (X\cup Y)}$ is a binary refinement
of $S'|_{\L(S')\setminus (X\cup Y)}$. Taking the latter two arguments together, 
$\hs$ is a binary refinement of $S'$.

We proceed with showing that  $\hs$ is a solution to $((T; t, \sigma), S')$. 
To this end, we apply Prop.\ \ref{prop:IFFbinRef} and show that 
$\hs$ agrees with $\RT (T; t,\sigma)$ and that $\ats{T}{\hs}$ is acyclic.

Let us first argue that $\hs$ agrees with $\RT (T; t,\sigma)$.  
Observe first that since $\hs$ contains $S^*|_X$ and $S^*|_Y$ as subtrees, $\hs$ displays all triples in $ab|c \in rt(S^*)$ 
with $a,b,c\in X$, or with $a,b,c\in Y$. Moreover, $\hs$ displays all triples $ab|c \in rt(S^*)$
for which at least one of $a,b$ and $c$ is not contained in $X\cup Y$.
The latter two arguments and $\lca_{\hs}(X\cup Y)=x$ imply that  
$\hs$ displays all triples  $ab|c \in rt(S^*)$ except possibly those 
for which $lca_{\hs}(a,b,c) = x$.
Let $R_x = \{ab|c \in rt(\hs) : lca_{\hs}(a,b,c) = x\}$.
By the latter arguments, the only triplets in $rt(\hs)$ that are not in $rt(S^*)$
are in $R_x$, i.e. $rt(\hs) \subseteq rt(S^*) \cup R_x$.
By the definition of a good split refinement, $S'$ agrees with $\RT (T; t,\sigma)$. 
Note that $R_x$ contains precisely those triples $ab|c$ for which either
$a,b\in X$ and $c\in Y$ or $c\in X$ and $a,b\in Y$. This observation immediately implies that 
$R_x \subseteq rt(S')$.
We thus have $rt(\hs) \subseteq rt(S^*) \cup rt(S')$ and since both $S^*$ and $S'$ agree with $\RT (T; t,\sigma)$, it follows that $\hat{S}$ agrees with $\RT (T; t,\sigma)$.

We must now argue that $\ats{T}{\hs}$ is acyclic.
Assume for contradiction that $\ats{T}{\hs}$ contains a cycle $C = (w_1, w_2, \ldots, w_k, w_1)$.  
Since $\hs$ is binary and agrees with $\RT (T; t, \sigma)$, $\hs$ displays $\RT (T; t, \sigma)$ and Theorem \ref{thm:SpeciesTriplets}
implies that there is a reconciliation map from 
from $(T;t,\sigma)$ to $\hs$. 
By Lemma \ref{lem:self-loops}, $\ats{T}{\hs}$ does not
contain self-loops and thus $k>1$ for $C = (w_1,\ldots, w_k, w_1)$.
We will derive a contradiction by showing that $\ats{T}{S^*}$ contains a cycle.
The proof is divided in a series of claims.

\begin{owndesc}
\item[Claim 1:]
\emph{If $(u, v) \in E(\ats{T}{\hs})$ and $u, v \notin V(\hs(x))$, then $(u, v) \in E(\ats{T}{S^*})$.}

Note that $\hs$ and $S^*$ are identical except for the subtree rooted at $x$.  
Thus, $(u, v)$ is an A2-edge in $E(\ats{T}{\hs}$ if and only if it is an A2-edge in $\ats{T}{S^*}$. 
Moreover, for all other edge types, we have $\hmu_{\hs}(u) = \hmu_{S^*}(u)$, $\hmu_{\hs}(v) = \hmu_{S^*}(v)$,
as well as $\lca_{\hs}(\hmu(u), \hmu(v)) = \lca_{S^*}(\hmu(u), \hmu(v))$.
This directly implies that every edge $(u, v) \in E(\ats{T}{\hs})$ that does not involve a vertex of $\hat{S}(x)$ is also in $S^*$.
This proves Claim 1. 
\end{owndesc}

To stress once again,
since $\hs$ is binary and agrees with $\RT (T; t,\sigma)$, it must display
    $\RT (T; t,\sigma)$. Thus, we can apply Theorem \ref{thm:SpeciesTriplets} to
    conclude that there is a reconciliation map from $\hs$ to $(T; t,\sigma)$.

Now, let $Z = (V(T) \cup V(\hs)) \setminus V(\hs(x)$.  
Observe that $Z = (V(T) \cup V(S^*)) \setminus V(S^*(x)$.
If $C$ does not contain a vertex of $V(\hs(x))$, then by Claim 1, 
every edge of $C$ is also in $\ats{T}{S^*}$.  Thus $C$ is also a cycle in $\ats{T}{S^*}$, contradicting that it is acyclic.  
Therefore, we may assume that $C$ contains at least one vertex from $V(\hs(x))$.  
On the other hand, assume that $C$ does not contain a vertex of $Z$.  Then all the vertices of $C$ belong to $V(\hs(x))$.  
Since, as we argued before, $\ats{T}{\hs}$ does not contain self-loops, we conclude that every edge $(u,v)$ of $C$ is either an A1- or an A2-edge of $\ats{T}{\hs}$
	that satisfies $v \prec_{\hs} u$. However, this implies that the edges of $C$ cannot form a cycle; a contradiction.
Therefore, $C$ must contain vertices from both $V(\hs(x))$ and $Z$.
Assume, without loss of generality, that $w_1 \in V(\hs(x))$ and $w_k \in Z$.

Now, $C$ can be decomposed into a set of subpaths that alternate between
vertices of $V(\hs(x))$ and of $Z$. More precisely, we say that a subpath $P =
(w_i, w_{i+1}, \ldots, w_l)$ of $C$, where $1 \leq i \leq l \leq k$, is a
$V(\hs(x))$-subpath if $w_i, \ldots, w_l \in V(\hs(x))$. Similarly, we say that
$P$ is a $Z$-subpath if $w_i, \ldots, w_l \in Z$. Now, $C = (w_1, \ldots, w_k)$ is a
concatenation of subpaths $P_1, P'_1, P_2, P'_2, \ldots, P_h, P'_h$ such that
for $1 \leq i \leq h$, $P_i$ is a non-empty $V(\hs(x))$-subpath and $P'_i$ is a
non-empty $Z$-subpath.

We want to show that  $\ats{T}{S^*}$ contains a cycle. To this end, we will construct
a cycle $C^*$ in $\ats{T}{S^*}$ such that $C^*$ is the concatenation of subpaths $P_1^*, P_1', \ldots, P_h^*, P_h'$, where each $P^*_i$ 
is a subpath of $\ats{T}{S^*}$ that replaces $P_i$.
First notice that for each $1 \leq i \leq h$, all the edges of $P'_i$ are in $\ats{T}{S^*}$ by Claim 1.  Therefore, every $P'_i$ is a path in $\ats{T}{S^*}$.

In what follows, we consider the $V(\hs(x))$-subpath $P_i = (w_p, w_{p+1}, \ldots, w_q)$, where $1 \leq i \leq h$ 
($w_p = w_q$ may be possible if $P_i$ consists of a single vertex only).
Notice that $w_{p - 1}$ and $w_{q+1}$ are in $Z$ (where we define $w_{p-1} = w_k$ if $p = 1$ and $w_{q+1} = w_1$ if $p = k$).
We construct a path $P^*_i = (w^*_1, \ldots, w^*_r)$ of $\ats{T}{S^*}$ such that 
$(w_{p-1}, w^*_1) \in E(\ats{T}{S^*})$ and $(w^*_r, w_{q+1}) \in E(\ats{T}{S^*})$.  

To this end, we provide the following 

\begin{owndesc}
\item[Claim 2:]
\emph{The vertex $x$ does not belong to $C$.}

Let $Q$ be a maximal topological sort of $\ats{T}{S}$
and let $Q'$ be a maximal topological sort of $\ats{T}{S'}$.
By Lemma~\ref{lem:q-stays-topo}, $\M(Q)  \subseteq \M(Q')$. 
Moreover, since $S'$ is a good split refinement, all the in-neighbors 
of $x$ in $\ats{T}{S'}$ belong to $Q$.  Since $\M(Q)  \subseteq \M(Q')$, all the in-neighbors of 
$x$ in $\ats{T}{S'}$ are also in $Q'$.
This and maximality of $Q'$ implies that
$x$ is itself also in $Q'$.
Let $\hat{Q}$ be a maximal topological sort of $\ats{T}{\hs}$.  
Since $\hs$ can be obtained from a sequence of split refinements starting from $S'$, 
Lemma~\ref{lem:q-stays-topo} implies that $\M(Q') \subseteq \M(\hat{Q})$. 
In particular, $x \in \M(\hat{Q})$.
Lemma \ref{lem:Qproperty} implies that $x$ cannot be contained in any cycle of $\ats{T}{\hs}$, 
which proves Claim 2.
\end{owndesc}

Recalling that $\ats{T}{\hs}$ does not contain self-loops,
every edge $(u, v)$ of $P_i$ is an A1- or A2-edge of $\ats{T}{\hs}$
and satisfies $v \prec_{\hs} u$. 
This implies that either $w_q \prec_{\hs} w_{q-1} \prec_{\hs} \ldots \prec_{\hs} w_p$, 
	 or that $w_p=w_q$. 
	In either case, we have  $w_q \preceq_{\hs} w_p$.
By Claim 2, $w_p \neq x$. This and $w_p\in V(\hs(x))$  implies that $w_p \prec_{\hs} x$.  
By construction of $\hs$ we therefore have
 $\L(\hs(w_p)) \subseteq X$ or $\L(\hs(w_p)) \subseteq Y$.  
We will assume, without loss of generality, that $\L(\hs(w_p)) \subseteq X$.
Since $w_q \preceq_{\hs} w_p$, we  have $\L(\hs(w_q)) \subseteq \L(\hs(w_p)) \subseteq X$.
We now construct two important sets $X_p \subseteq X$ and $X_q \subseteq X$ that are quite
helpful for out construction of a cycle $C^*$ in $\ats{T}{S^*}$.

\begin{owndesc}
\item[Claim 3:]
\emph{There exists a subset $X_p \subseteq X$ such that $w_p =\lca_{\hs}(X_p)$ and 
$(w_{p-1}, \lca_{S^*}(X_p)) \in E(\ats{T}{S^*})$.} 

Since $w_p \in V(\hs)$, the edge $(w_{p-1}, w_p)$ is either an A1-, A2- or
A3-edge in $\ats{T}{\hs}$. Suppose first that $(w_{p-1}, w_p)$ is an A2-edge. Then $w_{p-1}$ is
the parent of $w_p$ in $\hs$. Since $w_p \prec_{\hs} x$, this implies that
$w_{p-1} \in V(\hs(x))$, contradicting $w_{p-1} \in Z$. Therefore, this case is
not possible.

Suppose that $(w_{p-1}, w_p)$ is an A1-edge defined by some $(u, v) \in E(T)$.
Then $w_p\in V(\hs)$ implies  $w_p = \hmu_{\hs}(v) = \lca_{\hs}(\sT(v))$
 and we define $X_p = \sT(v)$.  We must prove that $(w_{p-1}, \lca_{S^*}(X_p)) \in E(\ats{T}{S^*})$.
Since $(u, v) \in E(T)$ yields the A1-edge $(w_{p-1}, w_p)$ in $\ats{T}{\hs}$, 
	  we have $t(v)\in \{\leaf, \spec\}$. 
		Hence, $(u, v)$ yields some A1-edge $(z, \lca_{S^*}(X_p))$ in $\ats{T}{S^*}$
	 for some vertex $z$. In what follows, we show that $z=w_{p-1}$.

If $w_{p-1} \in V(T)$, then $w_{p-1} = u$
and $(u, v)$ defines the A1-edge $(u, \hmu_{S^*}(v)) = (w_{p-1},
\lca_{S^*}(X_p))$ in $\ats{T}{S^*}$. If $w_{p-1} \in V(\hs)$, then $w_{p-1} =
\hmu_{S^*}(u)$. Since $w_{p-1} \in Z$, vertex $w_{p-1}$ must be a strict ancestor of $x$ in $\hs$. 
This and the fact that $S^*$ and $\hs$ coincide except possibly in 
		$S^*(x)$ and $\hs(x)$ implies that 
$\hmu_{\hs}(u) = \hmu_{S^*}(u) = w_{p-1}$.
Hence, $(w_{p-1}, \lca_{S^*}(X_p)) \in E(\ats{T}{S^*})$.

Finally, suppose that $(w_{p-1}, w_p)$ is an A3-edge defined by some $u \in
V(T)$. Then $w_{p-1} = u$ and $w_p = \hmu_{\hs}(u) = \lca_{\hs}(\sT(u))$, where
$\sT(u) \subseteq X$. Define $X_p = \sT(u)$. Then $(w_{p-1}, w_p) = (u,
\lca_{\hs}(X_p))$ and $(u, \hmu_{S^*}(u)) = (w_{p-1}, \lca_{S^*}(X_p)) \in
E(\ats{T}{S^*})$. This proves Claim 3. 
\end{owndesc}

\begin{owndesc}
\item[Claim 4:]
\emph{There exists a subset $X_q \subseteq X$ such that $w_{q} = \lca_{\hs}(X_q)$ and  
$(\lca_{S^*}(X_q), w_{q+1}) \in E(\ats{T}{S^*})$.}

We show first that $w_{q+1} \in V(T)$. Assume, for contradiction, that $w_{q+1} \in V(\hs)$. 
		Since $(w_q,w_{q+1})$ is an edge  of $\ats{T}{\hs}$ and 
		since  $w_{q} \in V(\hs)$, the edge $(w_q,w_{q+1})$ is an A2-edge in $\ats{T}{\hs}$.
	  However, this implies that  $w_{q+1} \prec_{\hs} w_q$ and thus, $w_{q+1}\in V(\hs(x))$;
		a contradiction to $w_{q+1} \in Z$. Hence, $w_{q+1} \in V(T)$. 
Therefore, $(w_q, w_{q+1})$ is either an A1- or A4-edge in $\ats{T}{\hs}$. 

Suppose first that $(w_q, w_{q+1})$ is an A1-edge of $\ats{T}{\hs}$ defined by some $(u, v) \in E(T)$.
Then $(w_q, w_{q+1}) = (\hmu_{\hs}(u), v)$, where $\hmu_{\hs}(u) = \lca_{\hs}(\sT(u))$ and where $\sT(u) \subseteq X$.  
Define $X_q = \sT(u)$. Then $(w_q, w_{q+1}) = (\lca_{\hs}(X_q), w_{q+1})$, and 
$(\hmu_{S^*}(u), v) = (\lca_{S^*}(X_q), w_{q+1})$ is an A1-edge of $\ats{T}{S^*}$.

Suppose instead that $(w_q, w_{q+1})$ is an A4-edge of $\ats{T}{\hs}$ defined by some $(u, v) \in \tredge_T$ with $u=w_{q+1}$.
Then $(w_q, w_{q+1}) = (\lca_{\hs}(\hmu_{\hs}(u), \hmu_{\hs}(v)), u) = (\lca_{\hs}(\sT(u) \cup \sT(v)), u)$. 
Define $X_q = \sT(u) \cup \sT(v)$. 
Hence,  $w_q = \lca_{\hs}(X_q)$,
and since $w_q = \in V(\hs(x))$, we must have $\sT(u) \cup \sT(v) \subseteq X$.
Moreover, $(\lca_{S^*}(\hmu_{S^*}(u), \hmu_{S^*}(v)), u) = (\lca_{S^*}(X_q), w_{q+1})$ is an A4-edge of $\ats{T}{S^*}$.
This completes the proof of Claim 4.
\end{owndesc}

\begin{owndesc}
\item[Claim 5:]
\emph{Let $X_p$ and $X_q$ be subsets of $X$ as defined in Claim 3 and 4.
Then in $\ats{T}{S^*}$, there exists a path from $\lca_{S^*}(X_p)$ to $\lca_{S^*}(X_q)$.}

By Claim 3 and 4 we have $w_p = \lca_{\hs}(X_p)$ and $\lca_{\hs}(X_q) = w_q$, respectively.
As argued after the proof of Claim 2,
we have $\lca_{\hs}(X_q) = w_q \preceq_{\hs} w_p = \lca_{\hs}(X_p)$.
Because $\hs$ contains $S^*|_X$ as a rooted subtree, it follows that 
$\lca_{S^*}(X_q) \preceq_{S^*} \lca_{S^*}(X_p)$.
Because of the A2-edges, there must be a path from $\lca_{S^*}(X_p)$ to 
$\lca_{S^*}(X_q)$ in $\ats{T}{S^*}$. This completes the proof of Claim 5.
\end{owndesc}

We may now finish the argument.
For each $1 \leq i \leq h$, we let $P^*_i$ be the path obtained from Claim 5. 
We claim that by concatenating the paths $P^*_1, P'_1, P^*_2, P'_2, \ldots, P^*_h, P'_h$ in $\ats{T}{S^*}$, 
we obtain a cycle.  
We have already argued that each $P^*_i$ and each $P'_i$ is a path in $\ats{T}{S^*}$.  
The rest follows from Claim 4, since it implies that for each $1 \leq i \leq h$, the last vertex of $P^*_i$ has the first vertex of $P'_i$ 
as an out-neighbor, and the last vertex of $P'_i$ has the first vertex of $P^*_{i+1}$ as an out-neighbor
(where $P^*_{h+1}$ is defined to be $P^*_1$).  
We have thus found a cycle in $\ats{T}{S^*}$, a contradiction to the acyclicity of $\ats{T}{S^*}$.

Hence, $\ats{T}{\hs}$ is acyclic. This and the fact that $\hs$ displays $\RT (T;t,\sigma)$
	implies that $\hs$ is a solution to $((T; t, \sigma), S')$. Therefore, $((T; t, \sigma), S')$ admits a solution.
\end{proof}

\begin{theorem}
	Algorithm \ref{alg:gtcRefinement} determines whether a given GTC instance 
	$((T; t, \sigma), S)$ admits a solution or not and, in the affirmative
	case, constructs a solution $S^*$ of $((T; t, \sigma), S)$. 
	\label{thm:algo1}
\end{theorem}
\begin{proof}
	Let $((T; t, \sigma), S)$  be GTC instance. 
	First it is tested in Line \ref{line:binary} whether $S$ is binary or not. 
	If $S$ is binary, then $S$ is already its binary refinement and 
	Prop.\ \ref{prop:IFFbinRef} implies that $S$ is a solution to 
  $((T; t, \sigma), S)$ if and only if $S$
	agrees with $\RT (T;t,\sigma)$ and $\ats{T}{S}$ is acyclic.  
	The latter is tested in Line \ref{line:GTC-properties}. 
	In accordance with 	Prop.\ \ref{prop:IFFbinRef}, the tree $S$ is returned whenever the latter conditions are satisfied and, 
	otherwise, 	``there is no solution'' is returned.
	
	Assume that $S$ is not binary. If $S$ admits no good split refinement, then 
	Alg.\ \ref{alg:gtcRefinement}	(Line \ref{line:no-good-split}) returns
	``there is no solution'', which is in accordance with Prop.\ \ref{prop:no-solution}. 
	Contrary, if $S$ admits a good split refinement $S'$, then we can apply 
	Theorem \ref{thm:equiv-refinement} to conclude that
	$((T; t, \sigma), S)$ admits a solution 
	if and only if $((T; t, \sigma), S')$  admits a solution at all.

	Now, we recurse on  $((T; t, \sigma), S')$ as new input of 
	Alg.\ \ref{alg:gtcRefinement}	in Line \ref{line:good-split}. 
	The correctness of Alg.\ \ref{alg:gtcRefinement} is finally ensured by 
	Theorem \ref{thm:equiv-refinement} which states that if 
	$((T; t, \sigma), S')$ admits a solution and thus, by Prop.\ \ref{prop:IFFbinRef}, 
	a binary refinement $S^*$ which is obtained by a series of good split refinements
	starting with $S$, is a solution for $((T; t, \sigma), S)$. 
\end{proof}

\subsection{Finding a Good Split Refinement}

To find a good split refinement, if any, we can loop through each cherry $x$ and ask 
``is there a good split refinement at $x$''?  Clearly, every partition $X_1,X_2$ of $\child(x)$
may provide a good  split refinement and thus there might be $O(2^{|\child(x)|})$ cases to be tested for each cherry $x$. 
To circumvent this issue, we define a second auxiliary graph that is an extension of the well-known Aho-graph
to determine whether a set of triplets is compatible or not \cite{semple2003phylogenetics,Steel:book,Aho:81}.
For a given set $R$ of triplets, the Aho-graph has vertex set $V$ and (undirected) edges $\{a,b\}$
for all triplets $ab|c\in R$ with $a,b,c\in V$. 
Essentially we will use this Aho-graph and add additionally edges to it.
The connected components of this extended graph eventually guides us to
the process of finding good split refinements. Before we make this definition 
more precise we give the following.

\begin{lemma}\label{lem:good-ancestors}
Let $Q$ be a maximal topological sort of $\ats{T}{S}$. If there exists a
good split refinement $S'$ of $S$ at a cherry $x$, then every strict ancestor of
$x$ in $S$ and $S'$ is in $\M(Q)$.
\end{lemma}
\begin{proof}
Let $S'$ be a good split refinement of $S$ at $x$.  
By construction, the sets of ancestors of $x$ in $S$ and $S'$ are equal.

Assume that there is a strict ancestor $y$ of $x$ that is not in $Q$.
Due to the A2-edges in $\ats{T}{S}$ there is a directed path $P$ from 
		 $y$ to $x$ in  $\ats{T}{S}$. 
Lemma \ref{lem:Qproperty} implies that none of the vertices along this path
$P$ are contained in $\M(Q)$. Since $y$ is a strict ancestor of $x$ in $S$, we 
can conclude that the parent $p(x)$ of $x$ in $S$ is not contained in $\M(Q)$. 	 
Again, due to the A2-edges of $\ats{T}{S'}$, the pair $(p(x), x)$ is an edge 
in  $\ats{T}{S'}$ and hence, $p(x)$ is an in-neighbor of $x$ in  $\ats{T}{S'}$. 
However, since $S'$ is a good split refinement of $S$,  
all the in-neighbors of $x$ in $\ats{T}{S'}$ must, by definition, 
belong to $\M(Q)$; a contradiction. Thus, every strict ancestor $y$ of $x$ in $S$ and $S'$
is in $\M(Q)$.
\end{proof}

In what follows, when we ask whether a fixed $x$ admits a good split
refinement, we can first check whether all of its ancestors are in $Q$, where $Q$ is
maximal topological sort of $\ats{T}{S}$. If this is not the case, then, by contraposition of
Lemma~\ref{lem:good-ancestors}, we may immediately conclude that there is no good split refinement at $x$.

Otherwise, we investigate $x$ further.  We define now the new auxiliary graph to determine whether the cherry $x$ of $S$ admits a good split refinement or not.  
\begin{definition}[Good-Split-Graph]
Let $(T;t,\sigma)$ be a gene tree and $S$ be a species tree. 
	Moreover, let $Q$ be a maximal topological sort of $\ats{T}{S}$.

We define  $G((T; t, \sigma), S, x) = (V, E)$ as the the undirected graph with vertex set
$V = \L(S(x))$. Moreover, an (undirected) edge $ab$ is contained in $E$ 
if and only if  $a,b \in \L(S(x))$ and $a,b$ are distinct and satisfy at least one of the following conditions:

\begin{description}
    \item[(C1)] 
    there exists $c \in \L(S(x))$ such that $ab|c \in \RT (T; t,\sigma)$;

    \item[(C2)]
    there exists an edge $(u, v) \in E(T)$ such that $t(u) \in \{\dup, \transfer\}$, $u \notin \M(Q)$, $t(v) = \spec$, and $\{a, b\} \subseteq \sT(v)$;

    \item[(C3)]
    there exists an edge $(u, v) \in E(T)$ such that $t(u) = t(v) = \spec$, $\hmu_{S}(u) = x$ and $\{a, b\} \subseteq \sT(v)$;

    \item[(C4)]
    there exists a vertex $u \in V(T) \setminus \M(Q)$ such that $t(u) \in \{\dup, \transfer\}$ and $\{a, b\} \subseteq \sT(u)$.
\end{description}
\label{def:good-split-graph}
\end{definition}

\begin{figure}[tbp]
	\begin{center}
		\includegraphics[width=.6\textwidth]{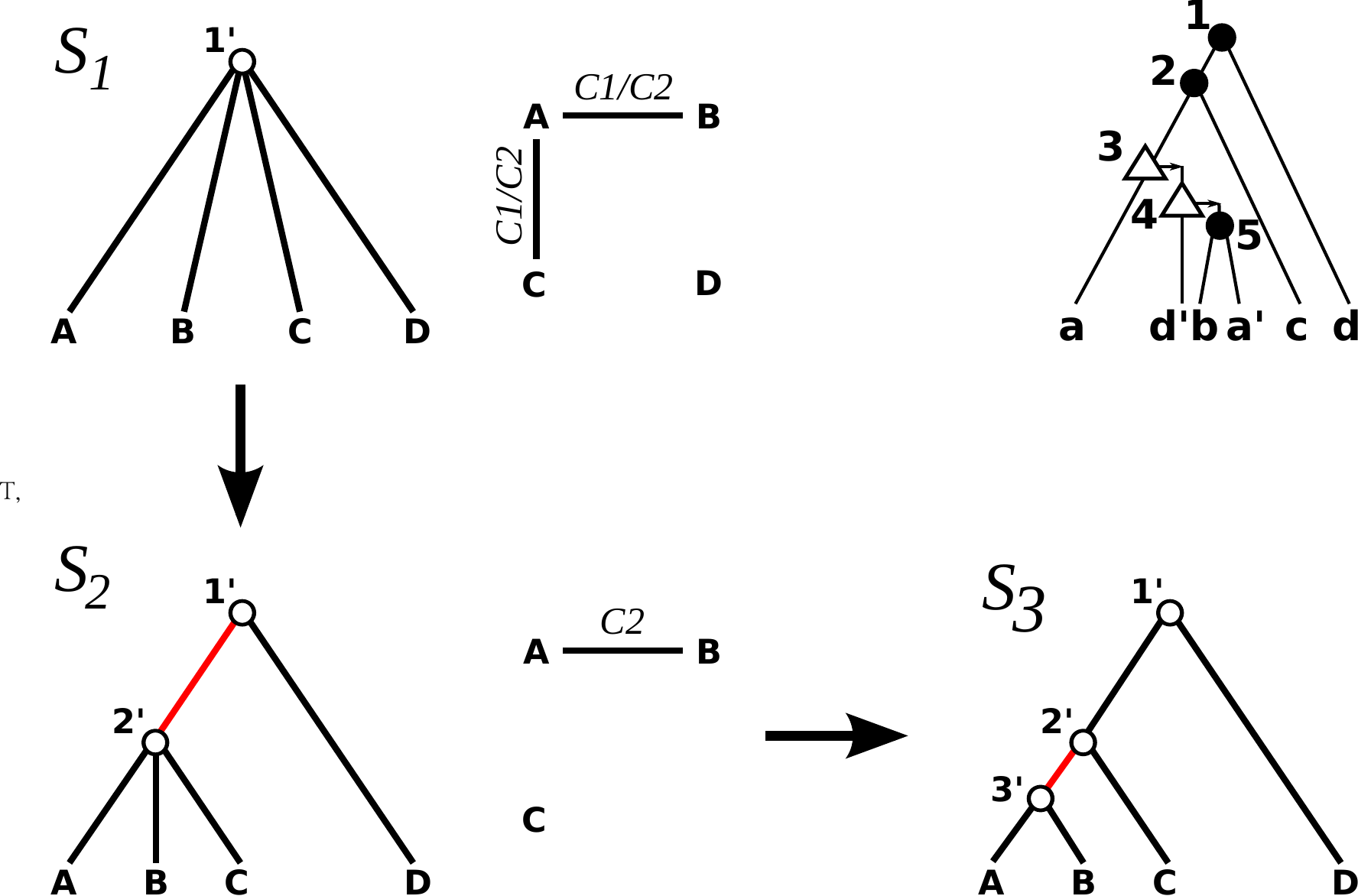}
	\end{center}
	\caption{Top right: the gene tree $(T;t,\sigma)$  from Fig.\ \ref{fig:least}	
		from which we obtain the species triples $\S(T;t,\sigma) = \{AB|D,AC|D\}$.
		We start with the star tree $S_1$ (top left) and obtain $G((T; t, \sigma), S_1, 1')$,
	 	which is shown right to $S_1$. $G((T; t, \sigma), S_1, 1')$ has four vertices $A,B,C,D$ and two edges. 
		The edge labels indicate which of the conditions in Def.\ \ref{def:good-split-graph}
		yield the respective edge. In $G((T; t, \sigma), S_1, 1')$, there is 
		only one non-trivial connected component which implies the good split that 
		results in the tree $S_2$ (lower left). There is only one cherry $2'$ in $S_2$ and 
 		the corresponding graph $G((T; t, \sigma), S_2, 2')$ is drawn right to $S_2$. 
		Again, the connected components give a good split that results in the binary 
		tree $S_3$. The tree $S_3$ is precisely the species tree as shown in the 
		middle of  Fig.\ \ref{fig:least}.
	}
\label{fig:working-exmpl-2}
\end{figure}

Intuitively, edges represent pairs of species that must belong to the same
part of a split refinement at $x$. That is, (C1) links species that would
contradict a triplet of $\RT (T; t, \sigma)$ if they were separated (as in
the classical BUILD algorithm \cite{semple2003phylogenetics,Steel:book,Aho:81}); 
(C2) links species that would yield an A1-edge from a
vertex not in $Q$ into $x$ if they were separated; (C3) links species that
would create a self-loop on $x$ if they were separated; and (C4) links
species that would create an A3-edge from a vertex not in $Q$ into $x$ if separated. We want the
graph to be disconnected which would allow us to split the children of $x$ while
avoiding all the situations in which we create a separation of two children where we cannot ensure that this separation yields
a good split refinement at $x$.  
Considering only such pairs of children turns out to be necessary and sufficient, and Theorem~\ref{thm:gsp-disco} below
formalizes this idea.

\begin{definition}
Given a graph $H$, we say that $(A, B)$ is a \emph{disconnected bipartition}
of $H$ if $A \cup B = V(H)$, $A \cap B = \emptyset$ and for each $a \in A, b \in B$, $ab \notin E(H)$.  
\end{definition}

We are now in the position to 
state how good split refinements can be identified.  
Note, we may assume w.l.o.g. that $S$ agrees with $\RT$, as otherwise there can be no good split refinement at all.

\begin{theorem}\label{thm:gsp-disco}
Let $((T; t, \sigma), S)$ be a GTC instance, and assume that $S$ agrees
with $\RT (T; t, \sigma)$. Let $Q$ be a maximal topological sort of
$\ats{T}{S}$. Then there exists a good split refinement of $S$ if and only
if there exists a cherry $x$ of $S$ such that every strict ancestor of $x$
in $S$ is in $Q$, and such that $G((T; t, \sigma), S, x)$ is disconnected. 

In particular, for any disconnected bipartition $(A, B)$ of $G$, the split refinement that partitions the children of $x$ into $A$ and $B$ is a good split refinement.
\end{theorem}

\begin{proof}
In  what follows, put  $G \coloneqq G((T; t, \sigma), S, x)$. 
Suppose that there exists a good split refinement $S'$ of $S$. 
Let $x$ be the cherry of $S$ that was refined from $S$ to $S'$, and let $x_1, x_2$ be the children of $x$ in $S'$. 
Let $Q$ be a maximal topological sort of $\ats{T}{S}$.  By Lemma~\ref{lem:good-ancestors}, every strict ancestor of $x$ in $S$ is in $Q$.
Let $A = \L(S(x_1))$ and $B = \L(S(x_2))$.
We claim that for any pair $a \in A, b \in B$, $ab \notin E(G)$.

Assume for contradiction that there is an edge $ab$ with $a \in A, b \in B$.  We treat each possible edge type separately.

\smallskip
\noindent
{\bf (C1)}: Suppose that $ab \in E(G)$ because there exists $c \in \L(S(x))$ such
that $ab|c \in \RT (T; t,\sigma)$. Because $a \in A$ and $b \in B$ and
by construction of $S'$, we either have $ac|b \in rt(S')$ if $c \in A$, or
$bc|a \in rt(S')$ if $c \in B$. In either case, $S'$ does not agree with
$\RT (T; t,\sigma)$, contradicting that $S'$ is a good split refinement.

\smallskip
\noindent
{\bf (C2)}: Suppose instead that $ab \in E(G)$ because there exists an edge $(u,
v) \in E(T)$ with $t(u) \in \{\dup, \transfer\}$, $u \notin \M(Q)$, $t(v) =
\spec$ and $a,b \in \sT(v)$. By construction of $S'$ and due to the choice
of $A = \L(S(x_1))$ and $B = \L(S(x_2))$, we have 
$\hmu_{S'}(v) = lca_{S'}(\sT(v))
\succeq_{S'} lca_{S'}(a,b) = x$. If $\hmu_{S'}(v) = x$, then $(u,
\hmu_{S'}(v)) = (u, x)$ is an A1-edge of $\ats{T}{S'}$.  Thus, $x$  has in-neighbor
$u$ in $\ats{T}{S'}$ such that  $u \notin
\M(Q)$, which contradicts that $S'$ is a good split refinement. So assume
that $\hmu_{S'}(v) \succ_{S'} x$. In this case, $(u, \hmu_{S'}(v))$ is an
A1-edge of $S'$, and by Lemma~\ref{lem:ancestors-dont-change}, $(u,
\hmu_{S'}(v)) \in E(\ats{T}{S})$. Since $u \notin \M(Q)$, we must have
$\hmu_{S'}(v) \notin \M(Q)$. Since $\hmu_{S'}(v) \succ_S x$, we obtain a
contradiction to Lemma~\ref{lem:good-ancestors}.

\smallskip
\noindent
{\bf (C3)}: Suppose that $ab \in E(G)$ because there exists an edge $(u, v) \in
E(T)$ with $t(u) = t(v) = \spec$, $\hmu_{S}(u) = x$ and $a,b \in \sT(v)$.
Note, since $\lca_S(a,b)=x$  and $a,b \in \sT(v)$, 
it must hold that $x\preceq_S\hmu_S(v)$. 
Moreover, $t(u) = t(v) = \spec$ implies that $u$ and $v$
are contained in the same connected component of $\Th$. This and  
$v\prec_{\Th} u$ implies $\sT(v)\subseteq \sT(u)$. Hence,
 $\hmu_S(v) \preceq_S \hmu_S(u)$.
Now,  $x\preceq_S\hmu_S(v)  \preceq_S \hmu_S(u) = x$ implies
 $\hmu_S(v) =\hmu_S(u) = x$. 
Therefore, $(\hmu_S(u), \hmu_S(v)) = (x, x)$ is an A1-edge
of $\ats{T}{S}$, and it follows that $x \notin \M(Q)$ (a vertex with a
self-loop cannot never be added to a maximal topological sort). Moreover,
because $a, b \in \sT(v)$ and $a \in A = \L(S(x_1))$ and $b\in B = \L(S(x_2))$, it holds that
$\hmu_{S'}(u) = \hmu_{S'}(v) = x$. Hence $(x, x)$ is an A1-edge of
$\ats{T}{S'}$ as well, and $x$ has an in-neighbor not in $Q$ (namely $x$
itself). This contradicts the assumption that $S'$ is a good split
refinement.

\smallskip
\noindent
{\bf (C4)}: 
Suppose that $ab \in E(G)$ because there is a vertex $u \in V(T) \setminus
\M(Q)$ such that $t(u) \in \{\dup, \transfer\}$ and $a,b \in \sT(u)$. The
reasoning is similar to Case (C2). That is, we must have $p :=
\hmu_{S'}(u) = lca_{S'}(\sT(u)) \succeq_{S'} lca_{S'}(a, b) =
x$. Now, $\ats{T}{S'}$ contains the A3-edge $(u, p)$. We cannot have $p =
x$ because $u \notin \M(Q)$ and $S'$ is a good split refinement of $S$. 
Thus $p \succ_{S'} x$. In this case, $(u, p)
\in E(\ats{T}{S})$ by Lemma~\ref{lem:ancestors-dont-change}. Thus $p$
cannot be in $\M(Q)$, which contradicts Lemma~\ref{lem:good-ancestors}.

\smallskip
We have thus shown that $ab$ cannot exist for any pair $a\in A$ and $b\in B$. 
Since $A$ and $B$ form a partition of $V(G)$, the graph $G$ must be disconnected.

Conversely, suppose that there exists a cherry $x$ of $S$ such that $G$ is
disconnected and such that every strict ancestor of $x$ in $S$ is in $Q$. 
Let $(A, B)$ be any disconnected bipartition of $G$.
Furthermore, let $S'$ be the split refinement of $S$ obtained by splitting
the children of $x$ into $A$ and $B$ and let $x_1, x_2$ be the two
children of $x$ in $S'$. W.l.o.g.\ assume that $x_1$ and $x_2$ is the ancestor of the 
leaves in $A$ and $B$, respectively. We claim that $S'$ is a good split
refinement. 

Let us first argue that $S'$ agrees with $\RT (T; t, \sigma)$.
Assume for contradiction that $S'$ displays a triplet $ac|b$, but that $ab|c \in \RT (T; t, \sigma)$.
By assumption, $S$ agrees with $\RT (T; t, \sigma)$, so $ac|b \in rt(S') \setminus rt(S)$.  
This implies that $\lca_{S'}(a, b) = \lca_{S'}(c, b) = x$.  
W.l.o.g\ we may assume that $a, c \in A$ and $b \in B$.  
However, Condition (C1) implies that we have the edge 
 $ab \in E(G)$, contradicting that $(A, B)$ forms a disconnected bipartition.  
Therefore, $S'$ agrees with $\RT (T; t, \sigma)$.

It remains to show that all in-neighbors of $x$ in $\ats{T}{S'}$ are contained in $\M(Q)$.  
Assume, for contradiction, that there is an edge $(p, x) \in E(\ats{T}{S'})$ 
such that $p \notin \M(Q)$.  
Since $x\in V(S')$, the edge $(p, x)$ it either an A1-, A2- or A3-edge in $\ats{T}{S'}$.  
As it is now our routine, we check several cases separately.

\begin{owndesc}
\item[\textnormal{\em Case: $(p, x)$ is an A1-edge and $p \neq x$.}] \ \\ 
In this case $(p, x)$ is defined by some edge $(u, v) \in E(T)$. Suppose
that $(p, x) = (\hmu_{S'}(u), \hmu_{S'}(v))$. Since $p \neq x$, $p$ is a
strict ancestor of $x$ in $S'$, and hence also in $S$. This is not
possible, since we assume that every strict ancestor of $x$ in $S$ belongs
to $Q$  (whereas here we suppose $p \notin \M(Q)$). We deduce that $(p, x) =
(u, \hmu_{S'}(v))$. Therefore, $u \notin \M(Q)$, $t(u) \in \{\dup,
\transfer\}$ and $t(v) = \spec$. Moreover, since $\hmu_{S'}(v) = x$
and $x$ has only the two children $x_1$ and $x_2$ in $S'$,
we can conclude there are $a,b \in \sT(v)$ such that $a \preceq_{S'}
x_1$ and $b \preceq_{S'} x_2$, i.e. $a \in A, b \in B$. 
The latter two arguments imply that Condition (C2) is satisfied for
$a$ and $b$ and, therefore, $ab\in E(G)$; a contradiction to $(A, B)$ are forming a
disconnected bipartition.

\item[\textnormal{\em Case: $(p, x)$ is an A1-edge and $p = x$.} ] \ \\
In this case, $(p, x) = (x, x) = (\hmu_{S'}(u), \hmu_{S'}(v))$ is defined by some edge $(u, v)$ of $T$.
Since $x$ is an internal vertex of $S'$, we must have $t(u) = t(v) = \spec$.
Since $L(S(x)) = L(S'(x))$ and $x$ is a cherry in $S$, we also have $(\hmu_S(u), \hmu_S(v)) = (x, x)$.
Moreover because $\hmu_{S'}(v) = x = \lca_{S'}(x_1,x_2)$, there must exist distinct $a,b$ with 
$a\prec_{S'}x_1$ and $b\prec_{S'}x_2$ such that $a,b \in \sT(v)$. 
Thus, $a \in A, b \in B$. Moreover $ab$ satisfies the Condition (C3). 
Thus, $ab\in E(G)$; a contradiction to our assumption that $(A, B)$ forms a
disconnected bipartition.

\item[\textnormal{\em Case: $(p, x)$ is an A2-edge.} ] \ \\
This case is not possible, since the parent of $x$ is the same in $S$ and $S'$, and we assume that all strict ancestors of $x$ in $S$ are in $Q$.

\item[\textnormal{\em Case: $(p, x)$ is an A3-edge.} ] \ \\
In this case, $(p, x) = (u, \hmu_{S'}(u))$ is defined by a vertex $u
\in V(T)$ such that $t(u) \in \{\dup, \transfer\}$ and $\hmu_{S'}(u) = x$. 
Since $u = p$ and, by assumption $p \notin \M(Q)$, we have $u \in V(T) \setminus \M(Q)$.
 
As in the A1-case, there must be $a, b \in \sT(u)$ such that $a \in A, b \in B$.  Then $ab$ should be an edge of $G$ because of Condition (C4), a contradiction.
\end{owndesc}

We have shown that the $(p, x)$ edge cannot exist. Therefore in
$\ats{T}{S'}$, all the in-neighbors of $x$ are in $Q$. Since $S'$ also
agrees with $R$, it follows that splitting the children of $x$ into $(A, B)$
forms a good split refinement at $x$.
\end{proof}

\begin{algorithm}[tbp]
\caption{\texttt{TimeConsistent Species Tree}}
\begin{algorithmic}[1]
  \Require Event-labeled gene tree $(T;t,\sigma)$ 
  \Ensure Time-consistent species tree $S$ for $(T;t,\sigma)$, if one exists
	\State Compute $\RT(T;t,\sigma)$ 
  	\State $S\gets $ star tree on $\sigma(L(T))$
  	\State Compute $\hmu_{T,S}(u)$ for all $u\in V(T)$  	
		\State Compute $\ats{T}{S}$
  	\State $Q \gets $maximal topological sort of $\ats{T}{S}$
  	\State Compute $G((T; t, \sigma), S, r)$, where $r$ is the root of $S$
	\State \texttt{Has\_GoodSplit} $\gets$ TRUE 
	\While{$S$ contains a non-binary cherry and \texttt{Has\_GoodSplit} = TRUE}		 
			\State \texttt{Has\_GoodSplit} $\gets$ FALSE 
			\ForAll{non-binary cherries $x$ of $S$ such that $y \in \M(Q)$ for all $y \succ_{S} x$} \label{for-loop}
					\If{\texttt{Has\_GoodSplit} = FALSE and $G((T; t, \sigma), S, x)$ is disconnected} \label{first-if}
								\State Compute disconnected bipartition $(A, B)$ of $G((T; t, \sigma), S, x)$. 
								\State $S \gets$ split refinement of $S$ at cherry $x$ based on $(A, B)$ 
							  	\State Compute  $\hmu_{T,S}(u)$ for all $u\in V(T)$
							  	\State Compute $\ats{T}{S}$
							  	\State $Q \gets $maximal topological sort of $\ats{T}{S}$
								\State Let $x_1, x_2$ be the children of $x$
								\State Compute $G((T; t, \sigma), S, x_1)$ and $G((T; t, \sigma), S, x_2)$
								\State \texttt{Has\_GoodSplit} $\gets$ TRUE
					\EndIf
			\EndFor
	\EndWhile
	\If{$S$ is binary}
	\ \Return $S$;
	\Else	 \ \Return ``No time-consistent species tree exists'';
	\EndIf
\end{algorithmic}
\label{alg:GoodSplit}
\end{algorithm}

A pseudocode to compute a time-consistent species for a given event-labeled gene tree $(T;t,\sigma)$, if one exists,
		is provided in Alg.\ \ref{alg:GoodSplit}.
		The general idea of Alg.\ \ref{alg:GoodSplit} is as follows. 
		With $(T;t,\sigma)$ as input, we start with a star tree $S$ and stepwisely refine $S$
		by searching for good split refinements. If in each step a good split refinement exists
		and $S$ is binary (in which case we cannot further refine $S$), then we found a time-consistent species tree $S$ for $(T;t,\sigma)$. 
		In every other case, the algorithm returns ``No time-consistent species tree exists''. 
		The correctness proof as well as further explanations are provided in the proof of Theorem \ref{thm:algo}. 
		To show that this algorithm runs in $O(n^3)$ time, we need first the following.

\begin{lemma}
$\RT (T; t, \sigma)$ can be computed in $O(n^3)$ time, where $n = |\L(T)|$.
This boundary is tight. 
\label{lem:comp-RT}
\end{lemma}
\begin{proof}
To compute $\RT (T; t, \sigma)$ as in Def.\ \ref{def:informativeTriplets} we can proceed as follows: 
We first compute the $\lca_{\Th}$'s for every pair of vertices within the connected components of $\Th$.
This task can be done in constant time for each pair of vertices after
linear preprocessing of the trees in $\Th$ \cite{HT:84,BF:00}.  
Thus, we end in an overall time complexity of $O(n^2)$ to compute all $\lca_{\Th}$'s between the leaves of $T$. 
We now compute the distance from the root $\rho_{\tilde T}$ to all other vertices in $V(\tilde T)$
for every connected component $\tilde T$ of $\Th$. The latter can be done for each individual connected 
component  $\tilde T$ via Dijkstra's algorithm in $O(|V(\tilde T)|^2)$ time. 
As this must be done for all connected components of $\Th$  and since 
$\sum_{\tilde T}  |V(\tilde T)|^2 \leq (\sum_{\tilde T}  |V(\tilde T)|)^2 = |V(T)|^2$
we end in time $O(|V(T)|^2) = O(n^2)$ to compute the 
individual distances. 
Now, for all three distinct leaves $a,b,c$ within the connected components of $\Th$, 
we compare the relative order of $x=\lca_{\Th}(a,b)$, $y=\lca_{\Th}(a,c)$, and $z=\lca_{\Th}(b,c)$ 
which can be done directly by comparing the distances $d_{\tilde T}(\rho_{\tilde T},x)$, $d_{\tilde T}(\rho_{\tilde T},y)$ and $d_{\tilde T}(\rho_{\tilde T},z)$. 
It is easy to see that at least two of the latter three distances must be equal. 
Hence, as soon as we have found that two distances are equal but distinct from the third, 
say $d_T(\rho_T,x)\neq d_T(\rho_T,y)=d_T(\rho_T,z)$, we found the triple $ab|c$ that is displayed by $\tilde T$.
If, in addition, $t(z)=\spec$ and $\sigma(a), \sigma(b), \sigma(c)$ are pairwise distinct, 
then we add $\sigma(a)\sigma(b)|\sigma(c)$ to  $\RT (T; t, \sigma)$. 
The latter tasks can be done in constant for every triple $a,b,c$.
Since there are at most $\binom{n}{3}=O(n^3)$ triplets in $T$, 
we end in an overall time-complexity $O(n^3)$ to compute all triplets displayed by $T$ 
that satisfy Def.\ \ref{def:informativeTriplets}(1).

Now we proceed to construct for all transfer edges $(u,v)\in \tredge_T$ the triplets $\sigma(a)\sigma(b)|\sigma(c)$
for all $a,b\in L(\Th(u))$ and  $c\in L(\Th(v))$ as well as for all $c\in L(\Th(u))$ and  $a,b\in L(\Th(v))$
with $\sigma(a), \sigma(b), \sigma(c)$ being pairwise distinct.
To this end, we need to compute $L(\Th(w))$ for all $w\in V(T)$. We may traverse
every connected component $\tilde T$ of $\Th$ from the root $\rho_{\Th}$ to each individual leaf and 
and for each vertex $w$ along the path from  $\rho_{\Th}$ to a leaf $l$, we add the leaf $l$ to  $L(\Th(w))$. 
As there are precisely $|L(\tilde T)|$ such paths, each having at most $|V(\tilde T)|\in O(|L(\tilde T)|)$ vertices, 
we end in $O(|L(\tilde T)|^2)$ time to compute  $L(\Th(w))$ for all $w\in V(\tilde T)$. 
As this step must be repeated for all connected components $\tilde T$ of $\Th$
we end, by the analogous arguments as in the latter paragraph, in $\sum_{\tilde T} O(|L(\tilde T)|^2) =  O(n^2)$ time
to compute $L(\Th(w))$ for all $w\in V(T)$. 
Now, for every transfer edge  $(u,v)\in \tredge_T$ the triplets $\sigma(a)\sigma(b)|\sigma(c)$ 
(with $\sigma(a), \sigma(b), \sigma(c)$ being pairwise distinct) are added to $\RT (T; t, \sigma)$
for all $a,b\in L(\Th(u))$ and  $c\in L(\Th(v))$ as well as for all $c\in L(\Th(u))$ and  $a,b\in L(\Th(v))$. 
Note, none of the trees $\Th$ contains transfer edges. Moreover, 
for each transfer edge $(u,v)$ we have, by Axiom (O3), $\sT(v)\cap \sT(u) = \emptyset$. 
The latter two arguments
imply that, for each transfer edge $(u,v)$, precisely $\binom{\sigma(L(\Th(v)))}{2} |\sigma(L(\Th(u)))|+ \binom{\sigma(L(\Th(u)))}{2} |\sigma(L(\Th(v)))|$ 
triplets are added.

Now, let $\mathcal{T} = \{T_1, T_2, \ldots, T_k\}$ be the set of trees in the forest 
$\Th$.  For each $i \in \{1, \ldots, k\}$, define $n_i = |\L(T_i)|$.  Let us write $T_i \rightarrow T_j$ if there exists a transfer edge $(u, v) \in \tredge_T$ satisfying $u \in V(T_i), v \in V(T_j)$.  It is easy to verify that there is exactly one transfer edge connecting
two distinct connected components of $\Th$, as otherwise, 
some vertex of some $T_j$ would have in-degree $2$ or more in $T$.
For each transfer edge $(u, v) \in \tredge_T$, where $u \in V(T_i)$ and $v
\in V(T_j)$, we can bound the number of added triplets by
$\binom{|\sT(u)|}{2} |\sT(v)|+ \binom{|\sT(v)|}{2}|\sT(u)| \leq n_i^2 n_j +
n_i n_j^2$. The total number of triplets considered is then
at most

\begin{align}
 \sum_{T_i \in \mathcal{T}} \sum_{T_j : T_i \rightarrow T_j} \left(n_i^2n_j + n_j^2n_i   \right)  
= & \sum_{T_i \in \mathcal{T}} n_i^2 \sum_{T_j : T_i \rightarrow T_j} n_j + 
\sum_{T_i \in \mathcal{T}} n_i \sum_{T_j : T_i \rightarrow T_j} n_j^2   \label{eq:App}\\
\leq & \sum_{T_i \in \mathcal{T}} n_i^2 \cdot n + \sum_{T_i \in \mathcal{T}} n_i \cdot n^2 \nonumber\\
= & n \sum_{T_i \in \mathcal{T}} n_i^2 + n^2 \sum_{T_i \in \mathcal{T}} n_i \nonumber\\
\leq & 2n^3 \in O(n^3) \nonumber.
\end{align}

In the latter approximation, we have used the fact that
distinct trees  $T_i$ and $T_j$ have  disjoint sets of leaf sets  (cf.\ \cite[Lemma 1]{nojgaard2018time}). Thus, 
$\sum_{T_i \in \mathcal{T}} n_i \leq n$ and $\sum_{T_j : T_i \rightarrow T_j} n_i \leq n$.
In summary, $\RT (T; t, \sigma)$ can be computed in $O(n^3)$ time.

Finally, note that if $(T; t, \sigma)$ is binary such that all inner vertices are labeled
as speciation $\spec$ and for all two distinct leaves $x,y\in L(T)$ we have $\sigma(x)\neq \sigma(y)$, 
then $|\RT (T; t, \sigma)|=\binom{n}{3}\in O(n^3)$. Hence, the boundary
$O(n^3)$ can indeed be achieved. 
\end{proof}

We note that it is not too difficult to show that Algorithm~\ref{alg:GoodSplit} can be implemented to take time $O(n^4)$.
Indeed, each line of the algorithm can be verified to take time $O(n^3)$, including the construction of $\ats{T}{S}$ 
(which takes time $O(n \log n)$, as shown in \cite[Thm.\ 6]{nojgaard2018time})
and the construction of the $G((T;t,\sigma), S, x)$ graphs (by checking every triplet of $\RT (T; t, \sigma)$ for (C1) edges, and 
for every pair $a,b$ of vertices, checking every member of $V(T) \cup E(T)$ for (C2), (C3) or (C4) edges).
Since the main \textit{while} loop is executed $O(n)$ times, this yields complexity $O(n^4)$.
However, with a little more work, this can be improved to cubic time algorithm.

As stated in Lemma~\ref{lem:comp-RT}, we may have $\RT (T;t,\sigma)
\in \Theta(n^3)$. Thus, any hope of achieving a better running time would
require a strategy to reconstruct a species tree $S$ without reconstructing
the full triplet set $\RT (T;t,\sigma)$ that $S$ needs to display. 
It may possible that 
such an algorithm exists, however, this would be a quite  
surprising result and may require a completely different approach.

\begin{theorem}\label{thm:algo}
Algorithm~\ref{alg:GoodSplit} correctly computes a time-consistent binary species tree for $(T;t,\sigma)$, if one exists, and
can be implemented to run in time $O(n^3)$, where $n = |\L(T)|$.  In particular, 
for every $O(n^3)$-time algorithm that needs to compute  $\S(T;t,\sigma)$ this boundary is tight. 
\end{theorem}

\begin{proof}
We first prove the correctness of the algorithm. Algorithm~\ref{alg:GoodSplit} takes as input a labeled gene tree $(T;t,\sigma)$. First $\S(T;t,\sigma)$ is computed
	and the star tree $S$ (which clearly agrees with $\S(T;t,\sigma)$) will be furthermore refined.
	Moreover, $S$ contains at this point of computation only one cherry, namely the root $r$ of $S$, and $G((T;t,\sigma),S,r)$
	is computed. Now, in each step of the \emph{while}-loop it is first checked if $S$ is non-binary and if in one of the 
	previous steps a good split refinement has been found. In this case, it is first checked (\emph{for}-loop)
	if there are non-binary cherries of the current tree $S$ for which all strict ancestors are contained in $\M(Q)$
	with $Q$ being the maximal topological sort of $\ats{T}{S}$. If this is not the case for all non-binary cherries, 
	the \emph{while}-loop terminates according to Lemma \ref{lem:good-ancestors} and the algorithm correctly outputs 
	\emph{``No time-consistent species tree exists''}. Contrary, if there is a non-binary cherry $x$ for which all strict ancestors are contained in $\M(Q)$, 
	then it is checked whether we have not found already a good split for $S$ and if $G((T;t,\sigma),S,x)$ is disconnected. 
	In this case, we can apply Theorem \ref{thm:gsp-disco} to conclude that there is a good split refinement for $S$ at $x$
	which is computed in the subsequent step. If, however, for all non-binary cherries $G((T;t,\sigma),S,x)$ is connected, 
	the algorithm correctly outputs 	\emph{``No time-consistent species tree exists''} according to Theorem \ref{thm:gsp-disco}.
	Finally, if in each step of the  \emph{while}-loop we have found a good split refinement and $S$ does not contain
	a non-binary cherry, then $S$ must, by construction, be binary. In this case, repeated application of Theorem \ref{thm:equiv-refinement}
	shows that the final binary tree $S$ is a solution to the underlying GTC instance and 
	the algorithm correctly returns $S$. Thus, Algorithm~\ref{alg:GoodSplit} correctly computes a time-consistent binary species tree for $(T;t,\sigma)$, if one exists.

We next analyze the running time of the algorithm.
Let $\Sigma = \sigma(\L(T))$ be the set of species.  We will frequently use the fact that $|\Sigma| \leq n$.
The main challenge in optimizing this algorithm is to be able to efficiently construct and update the $G((T;t,\sigma), S, x)$ graphs.  We will save this analysis for the end of the proof, and will ignore the time spent on graph updates for now.

We will assume that $\sT(u)$ is computed and stored for each $u \in V(T)$.  As argued in the proof of Lemma~\ref{lem:comp-RT}, 
this can be done in time $O(n^2)$.
Also by Lemma \ref{lem:comp-RT}, the triplet set $\RT (T; t, \sigma)$ can be computed in  $O(n^3)$ time.
Since every iteration of the main \textit{while} loop adds a new binary vertex
in $S$, the loop will be executed $O(n)$ times (since a binary tree on $|\Sigma| \leq n$
leaves has $O(n)$ internal vertices).  
By \cite[Lemma 3]{nojgaard2018time}, computing $\hmu_{T,S}$ can be done in time O($n\log(|\Sigma|))=O(n\log(n))$.  By \cite[Thm.\ 6]{nojgaard2018time}, the auxiliary graph $\ats{T}{S}$ can be computed in $O(|V(T)|\log(|V(S)|))$ time. 
Since $O(|V(T)|) = O(n)$ and $O(|V(S)|) = O(n)$, the latter task can be done in $O(n\log(n))$ time. 
Construction of $Q$ can be done in time
$O(|V(\ats{T}{S})| + |E(\ats{T}{S})|) = O(n)$ using the techniques of
Kahn~\cite{Kahn:62} and by observing that the edges in $E(\ats{T}{S})$ 
cannot exceed $|E(T)|+|E(S)|=O(n)$.

In each pass of the main \textit{while} loop, we iterate through $O(n)$
non-binary cherries. Let $c_1, \ldots, c_k$ be the non-binary cherries of
$S$, assuming that each auxiliary $G((T;t,\sigma), S, c_i)$ graph is
already pre-computed.  Since $c_1, \ldots,
c_k$ are cherries of $S$, the sets in $\{\L(S(c_1)), \L(S(c_2)), \ldots
\L(S(c_k))\}$ must be pairwise disjoint. Denoting $n_i = |\L(S(c_i))|$, $1
\leq i \leq k$, we thus observe that $\sum_{i = 1}^{k}n_i \leq n$. In the worst
case, we go through every cherry and check connectedness in time $O(n_i^2)$
on each graph $G((T;t,\sigma), S, c_i)$ via ``classical'' breadth-first search.  Thus in one iteration of the
main \textit{while} loop, the total time spent on connectedness
verification is $O(\sum_{i = 1}^{k}n_i^2) = O(n^2)$. When we apply a split
refinement, we compute $\hmu_{T,S}(u)$ for all $u\in V(T)$, $\ats{T}{S}$ and $Q$ at most once per
\textit{while} iteration, each operation being feasible in time $O(n \log n)$.
To be more precise, as soon as we have found a good split refinement we put 
\texttt{Has\_GoodSplit}=TRUE. Hence, the \textit{if}-condition (Line \ref{first-if}) will then not be satisfied, 
and we will not recompute the values $\hmu_{T,S}(u)$, $\ats{T}{S}$ and $Q$ 
again for the remaining non-binary cherries $x$ of $S$ within the \textit{for}-loop (Line \ref{for-loop}).
As there are $O(n)$ iterations, the time spent on operations other than
graph construction and updates is $O(n^3)$.

Let us now argue that the total time spent on the auxiliary $G((T;t,\sigma), S, x)$ graph updates can be implemented to take time $O(n^3)$.
To this end, we maintain a special data structure that, 
for each 2-element subset $\{a,b\} \subseteq \Sigma$, remembers the members of $\Sigma \cup V(T)$ that may cause $ab$ to be an edge in the auxiliary graphs.  
We describe this in more detail.  For a certain species tree $S$, we say that $a,b \in \L(S)$ are \emph{siblings} if $a$ and $b$ have the same parent in $S$.
For any two siblings $a,b$ of $S$, define
\begin{align*}
l_1(a,b) &= \{c \in \Sigma : c \mbox{ is a sibling of $a$ and $b$, and } ab|c \in \RT (T;t,\sigma)\} \\
l_2(a,b) &= \{u \in V(T) : \exists v \in V(T) \mbox{ such that $(u,v)$ satisfies (C2) }\} \\
l_3(a,b) &= \{u \in V(T) : \exists v \in V(T) \mbox{ such that $(u,v)$ satisfies (C3), with $x$ the parent of $a,b$ }\} \\
l_4(a,b) &= \{u \in V(T) : \mbox{$u$ satisfies (C4) }\}
\end{align*}

The $l_i$ sets are first initialized for $S$ being a star tree and then, subsequently updated 
	 based on the current refined tree $S$. 
		We will show below that the initial step to construct
		the $l_i$ sets for the star tree $S$ takes $O(n^3)$ time. 
		As we shall see below, when refining $S$ to $S'$ at some cherry $x$, only those elements remain in the $l_i$ sets 
				such that  (C1), (C2), (C3), resp., (C4) is satisfied for $a,b\in L(S(x))$  if and only if 
				$l_1(a,b), l_2(a,b), l_3(a,b)$, resp., $l_4(a,b)$ is non-empty.
Therefore, to decide whether there is an edge $ab$ in $G((T;t,\sigma), S, x)$ 
		it suffices to check whether one (of the updated) $l_i(a,b)$, $1\leq i\leq 4$
		 is non-empty, a task that can done in constant time for each two vertices $a,b\in \Sigma$.
		Hence, in each step we can construct $G((T;t,\sigma), S, x)$ in $O(n^2)$ time. 
	 	Thus, to show that the entire procedure runs in $O(n^3)$ time, we have to prove that we can update the $l_i$ sets in $O(n^3)$ time.

We show how to maintain these four sets for each pair of siblings as $S$ undergoes split refinements, starting with the initial star tree $S$.
The set $l_1(a,b)$ can be constructed in time $O(n^3)$ for all $a,b$ by iterating through $\RT (T;t,\sigma)$ once, and each $l_i(a,b)$, $i \in \{2,3,4\}$ can be constructed in time $O(n^3)$ by first constructing $\ats{T}{S}$ with its maximal topological sort $Q$ and, for each $a,b$ pair, checking every vertex and edge of $T$ for conditions (C2), (C3) and (C4).  It is easy to see that each condition can be checked in constant time per edge or vertex. 

Now assume inductively that $l_1,l_2,l_3$ and $l_4$ are known for each pair
of siblings of $S$. Instead of reasoning on the time to update these sets
during a particular iteration of the \textit{while} loop, we will argue on
the total number of times we must update each $l_i$ set during the whole
execution of the algorithm. Our aim is to show that each $l_i$ set requires
$O(n^3)$ updates in total, i.e. summing over all iterations of the loop.
When we apply a split refinement to $S$ at some cherry $x$, we create the
tree $S'$ on which we add the two children $x_1$ and $x_2$ under $x$. We
must update the corresponding sets. Let $X_1 = \L(S'(x_1))$ and $X_2 =
\L(S'(x_2))$. For $a,b \in X_1$, we may need to remove $c$ from $l_1(a,b)$
if $c \in X_2$ since it is not a sibling of $a$ and $b$ anymore. Thus after
a split refinement, for each $a,b \in X_1$ and each $c \in X_2$, we remove
$c$ from $l_1(a,b)$ if present (and we do the same for each $a'b' \in X_2,
c' \in X_1$). Therefore, each time that a pair $a,b \in \Sigma$ gets
separated from some $c \in \Sigma$ during the species tree construction, we
need $O(1)$ time to remove $c$ from $l_1(a, b)$. Importantly, this occurs
at most once during the whole algorithm execution. Therefore, in total we
spend time $O(1)$ one to update $l_1$ for each distinct $a,b,c \in \Sigma$,
and so the total time spent on updating $l_1$ is $O(n^3)$.

To update $l_2$ and $l_4$, we note that these two sets only depend on $t, \sT$
and $Q$, and only $Q$ is not fixed by the input. After a split refinement and
computing the new maximal topological sort $Q'$ of $\ats{T}{S'}$, by
Lemma~\ref{lem:q-stays-topo}, we only add new elements to $Q$ (i.e. if $\M(Q)
\subseteq \M(Q')$). Thus for each $q \in \M(Q') \setminus \M(Q)$, we must simply 
remove $q$ from $l_2(a, b)$ and $l_4(a,b)$, if present, for each $a,b \in \Sigma$.
This takes time $O(n^2)$ each time that a new vertex is added to a maximal topological sort after a split refinement.
Since, during the execution of the algorithm,
each vertex of $V(T) \cup V(S)$ is added to $Q$ at most once, this occurs
$O(n)$ times. It follows that the total time spend on updating $l_2$ and $l_4$ is $O(n^3)$. 

Finally, $l_3$ depends on $\hmu_{T,S}$ but not on $Q$.  After we apply a split refinement at $x$, transforming $S$ to $S'$,
$\hmu_S(u) \neq \hmu_{S'}(u)$ is only possible if $\hmu_{S}(u) = x$, in which case $\hmu_{S'}(u) \in \{x, x_1, x_2\}$.  For each $u$ such that $\hmu_{S}(u) = x$, we remove $u$ from $l_3(a,b)$ where necessary.  More precisely, if $\hmu_{S'}(u) = x$, we remove $u$ from $l_3(a,b)$ for each $a,b \in X_1 \cup X_2$ if present, since $x$ is not the parent of any two leaves $a,b$ now.  
If $\hmu_{S'}(u) = x_1$, we remove $u$ from $l_3(a,b)$ for each $a,b \in X_2$, and if $\hmu_{S'}(u) = x_2$, we remove $u$ from $l_3(a,b)$ for each $a,b \in X_1$.  One can see that for each $u \in V(T)$, we remove $u$ from $l_3(a,b)$ at most once for each distinct $a,b \in \Sigma$, and thus a total of $O(n^3)$ is spent on updating $l_3$ as well.

To summarize, the $l_i$ sets can be kept up-to-date after each split refinement in total time $O(n^3)$.
Since the other operations also tale time $O(n^3)$, the complete algorithm also takes $O(n^3)$ time.

Finally, among all algorithms that compute $\RT (T; t, \sigma)$, Lemma \ref{lem:comp-RT} implies that the boundary $O(n^3)$ is tight.
\end{proof}

\section{Summary and Outlook}

Here, we considered event-labeled gene trees $(T;t,\sigma)$ that contain speciation, duplication and HGT. 
We solved the Gene Tree Consistency (GTC) problem, that is,  
we have shown how to decide whether a time-consistent species tree $S$ for a given gene tree $(T;t,\sigma)$ exists 
and, in the affirmative case, how to construct such a binary species tree in cubic-time. 
Since our algorithm is based on the informative
species triples $\S(T;t,\sigma)$, for which  $\RT (T;t,\sigma)\in \Theta(n^3)$ may possible, there is no non-trivial way to
improve the runtime of our algorithm. Our algorithm heavily relies on the structure of an auxiliary graph $\ats{T}{S}$
to ensure time-consistency and good split refinements to additionally ensure that the final tree $S$ displays  $\S(T;t,\sigma)$.

This approach may have further consequence in phylogenomics. Since event-labeled gene trees $(T;t,\sigma)$
can to some extent directly be inferred from genomic sequence data, our method allows to test whether $(T;t,\sigma)$
is ``biologically feasible'', that is, there exists a time-consistent species tree for $(T;t,\sigma)$. 
Moreover, our method also shows that all information about the 
putative history of the species is entirely contained within the gene trees $(T;t,\sigma)$ and
thus, in the underlying sequence data to obtain  $(T;t,\sigma)$.

We note that the constructed binary species tree is one of possibly many other time-consistent species trees for $(T;t,\sigma)$. 
In particular, there are many different ways to choose a good split refinement, each choice may lead to a different species tree. 
Moreover, the reconstructed species trees here are binary. This condition might be relaxed and one may obtain further species tree by 
``contracting'' edges so that the resulting non-binary tree is still a time-consistent species tree for  $(T;t,\sigma)$. This eventually
may yield so-called ``least-resolved'' time-consistent species trees and thus,  trees that make no further assumption on 
the evolutionary history than actually supported by the data.

As part of further work, it may be of interest to understand in more detail, if our approach can be used to 
efficiently list all possible solutions, that is all time-consistent species trees for $(T;t,\sigma)$.

\bibliographystyle{plain}
\bibliography{main}

\end{document}